\documentclass[letterpaper]{article}

\usepackage{amsmath}
\usepackage{amssymb}
\usepackage{amsthm} 
\usepackage{xparse}
\usepackage{url}
\usepackage{authblk}
\usepackage{natbib}
\usepackage{fullpage}
\usepackage{color}
\usepackage{tikz}
\usetikzlibrary{decorations.pathreplacing}
\newtheorem{lm}{Lemma}
\newtheorem{theorem}[lm]{Theorem}
\newtheorem{lemma}[lm]{Lemma}

\newtheorem{claim}{Claim}
\newtheorem{observation}{Observation}

\newtheorem{definition}[lm]{Definition}

\usepackage[linesnumbered,noend,ruled,noline]{algorithm2e} 

\newcommand{\prof}{\mathbf{s}}
\newcommand{\load}{\ell}

\newcommand{\online}[1]{A(#1)}
\newcommand{\opt}{\text{OPT}}
\newcommand{\cp}{t}
\newcommand{\ve}{\varepsilon}
\newcommand{\high}{h}
\newcommand{\dis}{{enforcer}}
\newcommand{\diss}{{enforcers}}
\newcommand{\bounded}{bounded}
\newcommand{\disSet}{D}
\newcommand{\lastSeg}{\varkappa}
\newcommand{\Prob}{\mathbb{P}}
\newcommand{\CumCost}{Z}

\DeclareMathOperator*{\E}{\mathbb{E}}

\def \ve {\varepsilon}
\def \s {\mathbf{s}}

\begin{document}

\title{Resource-Aware Cost-Sharing Mechanisms with Priors}

\author[1]{Vasilis Gkatzelis\thanks{gkatz@drexel.edu}}
\author[1]{Emmanouil Pountourakis\thanks{manolis@drexel.edu}}
\author[2]{Alkmini Sgouritsa\thanks{alkmini@liv.ac.uk}}

\affil[1]{Drexel University}
\affil[2]{University of Liverpool}
\date{}

\maketitle

\begin{abstract}
In a decentralized system with $m$ machines, we study the selfish scheduling problem where each user strategically chooses which machine to use. Each machine incurs a cost, which is a function of the total load assigned to it, and some cost-sharing mechanism distributes this cost among the machine's users. The users choose a machine aiming to minimize their own share of the cost, so the cost-sharing mechanism induces a game among them. We approach this problem from the perspective of a designer who can select which cost-sharing mechanism to use, aiming to minimize the price of anarchy (PoA) of the induced games.

Recent work introduced the class of \emph{resource-aware} cost-sharing mechanisms, whose decisions can depend on the set of machines in the system, but are oblivious to the total number of users. These mechanisms can guarantee low PoA bounds for instances where the cost functions of the machines are all convex or concave, but can suffer from very high PoA for cost functions that deviate from these families. 

In this paper we show that if we enhance the class of resource-aware mechanisms with some prior information regarding the users, then they can achieve low PoA for a much more general family of cost functions. We first show that, as long as the mechanism knows just two of the participating users, then it can assign special roles to them and ensure a constant PoA. We then extend this idea to settings where the mechanism has access to the probability with which each user is present in the system. For all these instances, we provide a mechanism that achieves an expected PoA that is logarithmic in the expected number of users.
\end{abstract}

\section{Introduction}
In this paper we revisit a classic selfish scheduling problem: in a large decentralized system with a set $M$ of machines and a set $\mathcal{N}$ of registered users, each day some subset of these users enter the system seeking to process some task. Each user assigns their task to one of the machines, generating a cost that depends on the machine's total load, and the cost of each machine is then charged to its users, through some cost-sharing mechanism. The users' goal is to minimize their own share of the cost, so they strategically assign their task to the machine that would yield the smallest cost share. However, their cost share depends on the congestion of each machine, and thus on the strategic choices of all the other users currently in the system, giving rise to a game.

The need to better understand these games and to evaluate the efficiency of their outcomes lies at the heart of Algorithmic Game Theory, and some of the first seminal papers in this literature analyzed the \emph{price of anarchy} (PoA) of such games, i.e., the extent to which the performance of their Nash equilibria approximates the optimal performance. Much of this work, e.g., in congestion games and network formation games, assumed that the users share the cost \emph{equally}, in accordance with the Shapley value cost-sharing mechanism (e,g., see Chapters 18 and 19, respectively, from \cite{NRTV07-}). However, it soon became clear that the equal-sharing policy can lead to highly inefficient outcomes, even in very simple instances~\cite{ADKTWR08}. As a result, subsequent work focused on the design of alternative, more sophisticated, cost-sharing mechanisms, with the goal of reducing the PoA.

The first to study the extent to which a designer can reduce the PoA using improved cost-sharing mechanisms were \citet{CRV10}. One of their main goals was to analyze mechanisms that are \emph{stable} (i.e., guarantee the existence of pure Nash equilibria in the games they induce) and \emph{decentralized} (i.e., have limited information regarding the overall state of the system). Taking the need for decentralization to an extreme, they focused on the class of \emph{oblivious} cost-sharing mechanisms\footnote{Also known as \emph{uniform} mechanisms.}, which decide how to share the cost of each machine among its users without using \emph{any} information regarding the set of other users or machines that are present in the system. After providing a precise characterization of stable mechanisms for network formation games (where the resources that the agents use have \emph{constant} cost-functions), they systematically analyzed their performance. Building on this work, \citet{vFH13} and then \citet{CGS17} considered more general classes of cost functions. Among other results, \citet{vFH13} showed that no oblivious cost-sharing mechanism can guarantee a PoA bound better than linear function in the number of agents, even for instances with concave cost functions. 
Motivated by this limitation of oblivious mechanisms, subsequent work introduced the model of \emph{resource-aware mechanisms}~\cite{CS16,CGS17}. Compared to oblivious mechanisms, resource-aware ones are more informed: their decision regarding how to share the cost of a machine can also depend on the set of other machines that are available in the system. Using this additional information, \citet{CGS17} managed to overcome the limitations of oblivious mechanisms and design resource-aware mechanisms that achieve a constant PoA for convex and concave cost functions. On the negative side, they showed that there exists a class of seemingly simple cost functions for which no resource-aware cost-sharing mechanism can achieve a PoA better than $O(\sqrt{n})$.

These negative results suggest that it may be impossible for resource-aware mechanisms to achieve a constant PoA for interesting cost functions beyond convex and concave. However, although resource-aware mechanisms are more informed than oblivious ones, they are still severely limited in terms of what they know about the users in the system. In this paper we enhance resource-aware mechanisms with some prior information regarding the users in the system, and we show that this is sufficient for us to design cost-sharing mechanisms that achieve low PoA for a very broad class of cost functions.

\subsection{Our Results}\label{sec:results}
Our main results show that, using only a limited amount of prior information regarding the set of users in the system, resource-aware cost-sharing mechanisms can guarantee a low price of anarchy for a very wide class of selfish scheduling problems. 

\paragraph{Cost functions.} In contrast to prior work in cost-sharing mechanisms, which was mostly restricted to either convex or concave cost functions, our positive result applies to a much larger class of functions. Specifically, we consider any instance where the cost functions of the machines satisfy a mild condition regarding how fast they can grow. We call a cost function \emph{\bounded} if it satisfies the condition that $c(\ell+1)/c(\ell)=O(1)$ for all $\ell>0$, i.e., that the relative jump in the cost function can be upper bounded by some constant. Although the class of {\bounded} functions does not capture extreme examples of cost functions such as $c(\ell)=\ell^{\ell}$, where $c(\ell+1)/c(\ell)>\ell$, it captures the vast majority of functions that may characterize the cost incurred by some machine as a function of its load. For example, it includes all polynomial and even exponential cost functions. Note that this class also contains highly complicated functions that may not have a closed-form expression. 

\paragraph{Games with two known users}
We first consider resource-aware mechanisms that are oblivious to the set of users in the system, with the exception of just two users. The main idea behind our proposed mechanism is to assign special roles to these two users, referred to as \emph{\diss}, and carefully incentivize them to enforce an approximately efficient assignment in equilibrium. Using this approach we manage to guarantee a constant PoA for instances with any combination of \bounded\ cost functions.

\vspace{0.05in}
\noindent {\bf Theorem:} For any class of scheduling games with two known users and \bounded\ cost functions, there exists a stable resource-aware cost-sharing mechanism that achieves a constant PoA.
\vspace{0.05in}

\paragraph{Games with stochastic user arrivals.}
We then extend this idea to the case where each user enters the system with some probability $p$, and the mechanism knows $p$ but not the realization. To achieve a good PoA bound for this class of instances we assign the role of the \dis\ to more users, depending on the value of $p$, and guarantee an expected PoA that is at most logarithmic in the expected number of users.

\vspace{0.05in}
\noindent {\bf Theorem:} For any class of scheduling games with i.i.d.\ arrivals\footnote{We also extend this result to hold even if each bidder $i$ arrives with a different probability, $p_i$.} and \bounded\ cost functions,
there exists a stable resource-aware cost-sharing mechanism that achieves an expected PoA of $O(\log(\tilde{n}))$, where $\tilde{n}=p|\mathcal{N}|$ is the expected number of users.
\vspace{0.05in}

\paragraph{Technical obstacles}
Designing efficient cost-sharing mechanisms for such a wide family of instances is quite demanding: the structure of the optimal assignment can change, depending on the actual number, $n$, of users in the system, but the mechanism is oblivious to this number. So, how can the mechanism approximate the optimal solution without knowing $n$? Prior work focused on the case of concave or convex cost functions, and designed mechanisms leveraging the fact that the corresponding optimal assignments are reasonably ``well behaved'': for concave costs there always exists an optimal solution where all the jobs are assigned to a \emph{single} machine, and for convex costs an optimal assignment can be reached using a simple \emph{greedy} solution~(e.g., \cite{CGS17,CGLS20}). However, we cannot expect to find such convenient structural properties when dealing with the vast family of \bounded\ functions, because the optimal assignment can change radically as a function of $n$.

To deal with this fundamental obstacle, we propose a novel solution: rather than trying to implement the optimal assignment in equilibrium, we instead seek to implement a ``well behaved'' alternative assignment implied by an \emph{online algorithm}. This algorithm assigns jobs to machines using a predetermined assignment sequence which is \emph{independent of the total number of jobs, $n$}. We prove that this algorithm has a constant competitive ratio and then carefully design our cost-sharing mechanisms aiming to implement the outcome of this algorithm in equilibrium, thus inheriting a good approximation guarantee. We believe this technique may be of independent interest.

\subsection{Related Work}\label{sec:related}
Our work extends the recent literature that uses resource-aware cost-sharing mechanisms to achieve low PoA in different classes of games. \citet{CS16} were the first to study this family of mechanisms\footnote{In their paper, \citet{CS16} refer to these mechanisms as \emph{universal} instead of resource-aware.}, focusing on the class of network formation games (like \citet{CRV10} did for the family of oblivious mechanisms). Unlike the scheduling games that we study in this paper, network formation games take place over a graph: each agent is associated with a vertex of the graph and needs to use a path connecting that vertex to a designated sink-vertex, and each edge of the graph corresponds to a resource with a constant cost function. One can think of the games in our paper as the special case where the graph has just two vertices (a source and a sink) and several parallel edges: each edge corresponds to one of the machines, and every agent needs to choose one of these edges in order to get from the source to the sink. From this perspective, our games are more restricted in terms of the users' strategy space, but quite more general in terms of the classes of cost functions. \citet{CS16} showed that when the graph is outerplanar, then resource-aware mechanisms can outperform oblivious ones, but they also proved that an analogous separation is not possible for general graphs. In subsequent work, \citet{CGS17} designed resource-aware mechanisms for the same class of scheduling games that we study in this paper, and were able to achieve a constant PoA for instances with convex and concave cost functions. In a recent paper, \citet{CGLS20} extended many of these results to more general graphs, beyond parallel links, including directed acyclic or series parallel graphs with convex or concave cost functions on the edges.

The assumption that the cost-sharing mechanism may have additional prior information regarding the users was also part of the model studied by \citet{CS16} for the case of network formation games. Specifically, rather than assuming that the source vertex of each agent is chosen adversarially, they assumed that it is drawn from a distribution over all vertices. The cost-sharing mechanism is aware of this stochastic process, so they designed a mechanism that leverages this information to achieve a constant PoA. Following-up on this work~\citet{CLS19} extended the constant PoA to also include Bayesian Nash equilibria.

Ensuring that a cost-sharing mechanism is stable can be quite demanding, so characterizations of stable mechanisms can be very useful. Building on the impressive characterization of stable oblivious mechanisms by \citet{CRV10}, \citet{GMW14} provided a characterization for the set of stable oblivious cost-sharing mechanisms. They proved that these mechanisms correspondend to the class of generalized weighted Shapley values. Leveraging this characterization, \citet{GKR16} analyzed this family of cost-sharing protocols and showed that the unweighted Shapley value achieves the optimal price of anarchy guarantees for a large family of network cost-sharing games.

Other papers on the design and analysis of cost-sharing protocols include \citet{HF14}, who focused on capacitated facility location games, \citet{MW13} who considered a utility maximization model, and \citet{hhss21}, who considered a model that imposes some constraints over the portions of the cost that
can be shared among the agents. Also, \citet{HM11} 
studied the performance of several cost-sharing protocols in a setting, where each player can declare a different demand for each resource. 

Finally, there are several other models in which cost-sharing has played a central role. For example, \citet{MS01} focused on participation games, while \citet{M08} and \citet{MR09} studied queueing
games. \citet{CGV17} recently also pointed out some connections between cost-sharing mechanisms and the literature on coordination mechanisms, which started with the work of \citet{CKN09} and led to several papers focusing on decentralized scheduling policies for machine scheduling games \cite{I+09, AJM08, C09, AH12, K13, C+13, CMP14, BIKM14}. Just like the research on cost-sharing mechanisms, most of the work in coordination mechanisms studies how the price of anarchy varies with the choice of local scheduling policies on each machine (i.e., the order in which to process jobs assigned to the same machine).

\section{Preliminaries}\label{sec:prelim}
We analyze the scheduling games that arise in a decentralized system with a set $M=\{1, 2, \dots , m\}$ of $m$ machines and a set $N=\{1, 2, \dots , n\}$ of $n$ users. 
Each user owns a job and needs to schedule it on one of the machines. Each machine $j\in M$ is characterized by a cost function $c_j: \mathbb{N}\to \mathbb{R}$, where $c_j(\ell)$ is the cost that the machine would incur for processing a total of $\ell$ jobs. The cost function satisfies $c_j(0)=0$ and it is non-decreasing. 

The strategy profile, $\prof = (s_1, s_2, \dots , s_{n})$, of a scheduling game is a schedule, where 
$s_i$ corresponds to the machine that player $i$ chooses for her job. We use $S_j(\prof)=\{i\in N:s_i=j\}$ 
to denote the set of players who scheduled their jobs on machine $j$ in profile $\prof$, and 
$\load_j(\prof)=|S_j(\prof)|$ to denote the load on machine $j$ in $\prof$. Therefore, the cost of machine 
$j$ in this schedule is $c_j(\load_j(\prof))$, and the overall generated cost is $C(\prof)=\sum_{j\in M}c_j(\load_j(\prof))$. 
For notational simplicity, apart from $c_j(\load_j(\prof))$, we also use $c_j(\prof)$ to denote the cost of machine $j$ in $\prof$, since $j$'s load is directly implied by $\prof$.

{\bf Cost-Sharing Mechanisms.}
A \emph{cost-sharing mechanism} is a protocol that determines the cost of each agent using a machine. Formally, a cost-sharing mechanism~$\Xi$ defines at each schedule $\prof$ a 
nonnegative cost share $\xi_{ij}(\prof)$ for each $j\in M$ and $i\in S_j(\prof)$. Since the 
machine that $i$ uses, i.e., $s_i$, is implied by $\prof$, we also 
denote this cost share as $\xi_i(\prof)$. In any schedule $\prof$, the cost of each machine $j$ must be fully covered by the agents using it, so $\sum_{i\in S_j(\prof)} \xi_i(\prof)\geq c_j(\load_j(\prof))$. We use $\hat{C}(\prof)=\sum_{j\in M}\sum_{i\in S_j(\prof)} \xi_i(\prof)$ to denote the overall cost suffered by the users in $\prof$. Since the agents pay at least the cost they generate, we have $\hat{C}(\prof)\geq C(\prof)$ for every profile $\prof$. If there exists some profile for which this inequality is strict, i.e., the cost suffered by the users is greater than the cost that they generated, then we say that the cost-sharing mechanism uses \emph{overcharging}.

{\bf Resource-Aware Mechanisms.} In the class of \emph{resource-aware} cost-sharing mechanisms,
the value of the cost-share $\xi_{ij}(\prof)$ for each $i\in S_j(\prof)$ can depend on the set $S_j(\prof)$
of agents using that machine, on the set of machines $M$, and their cost functions, but not on the set 
$N\setminus S_j(\prof)$ of agents using other machines. In this paper we enhance this class of mechanisms with
some prior stochastic information regarding the set $N\setminus S_j(\prof)$, which enriches the set of cost-sharing
functions that we can implement, allowing us to achieve improved performance guarantees.

{\bf Pure Nash Equilibrium (PNE).} A tuple $(N, M, \mathbf{c}, \Xi)$ of a set of agents, a set of machines and their cost functions $\mathbf{c}=(c_j)_{j\in M}$, and a cost-sharing mechanism, defines a scheduling game $G$. The goal of every user in this game is to choose a machine that minimizes her own share of the cost, determined by $\Xi$.
A strategy profile~$\prof$ is a \emph{pure Nash equilibrium} (PNE) of this game $G$ 
if for every player $i\in N$, and every strategy $s'_i\in M$
\begin{equation*}
\xi_i(\prof)~=~\xi_i(s_i, \prof_{-i}) ~~\le~~ \xi_i(s'_i, \prof_{-i}),
\end{equation*}
where $\prof_{-i}$ denotes the profile of strategies for all agents other
than $i$. In other words, in a PNE $\prof$ no agent can decrease her cost share by 
unilaterally deviating from machine $s_i$ to $s'_i$ if all the other agents' 
choices remain fixed. 

{\bf Stability.} In accordance with prior work, we restrict our attention 
to \emph{stable} cost-sharing mechanisms, i.e., ones that induce games possessing at least one PNE.

{\bf Price of Anarchy (PoA).} 
To measure the performance of a cost-sharing mechanism in a given game, $G$, 
we evaluate the total cost $\hat{C}(\prof)$ suffered by the users in the worst equilibrium $\prof$, 
and compare it to the minimum total cost they could suffer. If we let $Eq(G)$ be the set of all PNE of $G$ and $F(G)$ denote the set of all its feasible schedules, then the \emph{price of anarchy} (PoA) of game $G$ is
\[\text{PoA}(G)=\frac{\max_{\prof\in\textrm{Eq($G$)}}\hat{C}(\prof)}{\min_{\prof^*\in F(G)}C(\prof^*)}\, .\]

Rather than evaluating the performance of cost-sharing mechanisms on a single game, we evaluate them on large classes of games. A class of scheduling games, $\mathcal{G}$, is defined by a tuple $({\mathcal N}, {\mathcal C}, \Xi )$, 
which comprises a universe of players $\mathcal N$, a universe of cost functions $\mathcal C$, 
and a cost sharing mechanism $\Xi$. An instance of a scheduling game $G\in \mathcal G$ 
consists of some subset of users $S\subseteq \mathcal N$, a set $M$ of machines with cost 
functions from $\mathcal C$, and the cost sharing mechanism $\Xi$. The \emph{worst-case
price of anarchy} of mechanism $\Xi$ for a class of games $\mathcal{G}$ is then defined as 
$\text{PoA}(\mathcal G)=\sup_{G\in \mathcal G} \text{PoA}(G)$. We also consider
settings where the subset of agents, $S$, is drawn from $\mathcal N$ based on some distribution $P$.
In that case, we evaluate the \emph{expected price of anarchy} of $\Xi$ as 
$$\text{ExpectedPoA}(\mathcal G)=\sup_{M, \mathbf{c}\in \mathcal C^{|M|}}\left\{\E_{S\sim P} [\text{PoA}((S,M, \mathbf{c},\Xi))]\,\right\}.$$ 
In other words, given an adversarial choice of machines $M$ using cost functions from $\mathcal C$, we evaluate the expected PoA over the randomness of $P$ in defining the subset of agents $S$.

{\bf Classes of Cost Functions.} We say that a cost function is \emph{\bounded} if $c(\ell+1)/c(\ell)=O(1)$ for all $\ell>0$. Another class of functions that plays an important role in prior work is that of \emph{capacitated constant} cost functions. That is, functions such that $c(\load) = c$ when 
$\load\leq \cp$ and $c(\load)=\infty$ when $\load>\cp$, for some positive constants $c$ and $\cp$.  Note that, although these cost functions are not \bounded, one of our first results shows that we can achieve a small PoA for them as well, as long as their capacity, $\cp$, is at least 4. Finally, a \emph{4-step} function is a step function whose segments have length at least 4. In other words, the value of a 4-step function does not change more than once within any interval of length 4 in its domain. Note that capacitated constant functions with capacity at least 4 are a special case of a 4-step function. Also, it is easy to verify that for any \bounded\ cost function $c'$, there exists a 4-step function $c$ such that $c(\ell)\geq c'(\ell)$ and $c(\ell)/c'(\ell)=O(1)$ for all $\ell>0$\footnote{To verify this fact, note that given a \bounded\ function $c'$, we can define a 4-step function $c$ such that for every $k\in \mathbb{N}$, if $\ell\in [4k-3, 4k]$ then $c(\ell)=c'(4k)$. Clearly, $c(\ell)\geq c'(\ell)$ for all $\ell>0$. Also, since $c'$ is \bounded, this means that for every $\ell$ we have $c(\ell)/c'(\ell)\leq c'(\ell+4)/c'(\ell)=O(1)$.}. 
This means that we can always approximate a \bounded\ cost function using a 4-step function, so in the rest of the paper we assume that the cost functions are all 4-step functions.

{\bf Global Ordering.} Our mechanisms, as well as many mechanisms in the related work 
(e.g. \cite{Mou99,CGS17,CGLS20}), use a {\em global ordering} $\pi$ over the universe $\mathcal N$ of players in deciding how to distribute the cost. Although the externality of the users in the games that we study is symmetric (e.g., they all cause the same marginal increase in the cost of a machine), the mechanism needs to share the cost unevenly among them to achieve a good PoA\footnote{It is well known that the PoA is linear in the number of agents if we share the cost equally~\cite{ADKTWR08}.}. The global ordering provides a consistent way for the mechanism to differentiate between these users. To ensure that no fairness concerns arise from the asymmetry introduced by these mechanisms, we assume that this global ordering can change periodically in a predetermined way, thus providing a symmetric treatment of the users over time.

\section{Online Scheduling Algorithm}
The main obstacle that resource-aware mechanisms face in approximating the optimal solution is that they do not know the number $n$ of agents that are present in the system. 
Since the optimal solution can change radically as a function of $n$, how can the cost-sharing mechanism try to approximate it without knowing the value of $n$? 

Rather than trying to implement the optimal assignment as an equilibrium, the main idea behind our solution is to instead implement a much more ``well behaved'' allocation that, in turn, closely approximates the cost of the optimal assignment. Specifically, we define an online algorithm, called \textsc{Delayed-OPT}, which sequentially assigns jobs to machines using a predetermined order, without knowing the value of $n$. We show that this algorithm has a constant competitive ratio and then we design cost-sharing mechanisms aiming to implement the outcome of this algorithm in equilibrium.

If $\online{n}$ is the outcome of the online algorithm and $\opt(n)$ is the optimal allocation (i.e. the feasible schedule with the minimum social cost) when the total number of jobs is $n$, then the competitive ratio is equal to $\max_n \{C(\online{n})/C(\opt(n))\}$.
To simplify the description of the \textsc{Delayed-OPT} online algorithm, without loss of generality we normalize the costs functions. That is, all costs are multiplied by the same constant such that the minimum non-zero cost is equal to $1$.For each $k\in \mathbb{N}$, let $a_k=\max\{q\in \mathbb{N} :~ C(\opt(q))<2^k\}$ be the largest 
number of jobs such that the optimal social cost for scheduling these jobs remains less than $2^k$ (Figure~\ref{fig:onlineExample} 
shows two examples for capacitated constant cost functions). 
Using this definition, let $\ell_{jk}^*$ denote the number of jobs assigned to machine $j$ in the optimal allocation when the total number of jobs is $a_k$.

When the $q^\text{th}$ job arrives, the \textsc{Delayed-OPT} finds the smallest value of $k$ such that for some machine $j\in M$ the number of jobs, $\ell_j$, assigned to it so
far is less than $\ell_{jk}^*$. Then, among all such machines, the algorithm assigns this job to
the one that has the smallest index\footnote{We assume that the machines have some arbitrary, but fixed, ordering indicated by their indices.}. The algorithm then increments the value of $\ell_j$ by one and moves on to the next job. A formal description of the \textsc{Delayed-OPT} algorithm is provided as Algorithm~\ref{alg:DelayedOPT}, below, and two examples of the induced assignment are provided in Figure~\ref{fig:onlineExample}.

\begin{figure}[htb]
	\centering
	\begin{tikzpicture}
		
	\draw[line width=1pt, green!40!black] (-0.4,5.4) -- (4.6,5.4) -- (4.6,-0.2) -- (-0.4, -0.2) -- (-0.4,5.4);
	\draw[line width=1pt, green!40!black] (4.6,5.4) -- (11.6,5.4) -- (11.6,-0.2) -- (4.6, -0.2) -- (4.6,5.4);

	\draw (0,0) rectangle (0.6,0.6);
	\draw[line width=1pt] (0,0.6) -- (0.6,0.6);
	\node at (0.3,0.8) {$c_1 =1$};
	\node at (0.3,0.3) {$1$};
		
	\draw (1,0) rectangle (1.6,1.8);
	\draw[line width=1pt] (1,1.8) -- (1.6,1.8);
	\node at (1.3,2) {$c_2=2$};
	\node at (1.3,0.3) {$2$};
	\node at (1.3,0.9) {$3$};
	\node at (1.3,1.5) {$4$};

	\draw (2,0) rectangle (2.6,3.6);
	\draw[line width=1pt] (2,3.6) -- (2.6,3.6);
	\node at (2.3,3.8) {$c_3=15$};
	\node at (2.3,0.3) {$5$};
	\node at (2.3,0.9) {$6$};
	\node at (2.3,1.5) {$7$};
	\node at (2.3,2.1) {$8$};
	\node at (2.3,2.7) {$9$};
	\node at (2.3,3.3) {$10$};
	
	\draw[line width=1pt,dashed, green!40!black] (3,0) -- (3,4);
	
	\node at (3.8,0.5) {$a_0=0$};
	\node at (3.8,1.1) {$a_1=1$};
	\node at (3.8,1.7) {$a_2=4$};
	\node at (3.8,2.3) {$a_3=4$};
	\node at (3.8,2.9) {$a_4=6$};
	\node at (3.8,3.5) {$a_5=10$};
	
	\node at (2.1,-1) {(a)};

	\draw (5,0) rectangle (5.6,0.6);
	\draw[line width=1pt] (5,0.6) -- (5.6,0.6);
	\node at (5.3,0.8) {$c_1 =1$};
	\node at (5.3,0.3) {$1$};
	
	\draw (6,0) rectangle (6.6,1.2);
	\draw[line width=1pt] (6,1.2) -- (6.6,1.2);
	\node at (6.3,1.4) {$c_2 =2$};
	\node at (6.3,0.3) {$2$};
	\node at (6.3,0.9) {$3$};
	
	\draw (7,0) rectangle (7.6,1.2);
	\draw[line width=1pt] (7,1.2) -- (7.6,1.2);
	\node at (7.3,1.4) {$c_3 =2$};
	\node at (7.3,0.3) {$12$};
	\node at (7.3,0.9) {$13$};

	\draw (8,0) rectangle (8.6,1.2);
	\draw[line width=1pt] (8,1.2) -- (8.6,1.2);
	\node at (8.3,1.4) {$c_4 =2$};
	\node at (8.3,0.3) {$14$};
	\node at (8.3,0.9) {$15$};

	\draw (9,0) rectangle (9.6,4.8);
	\draw[line width=1pt] (9,4.8) -- (9.6,4.8);
	\node at (9.3,5) {$c_5=7$};
	\node at (9.3,0.3) {$4$};
	\node at (9.3,0.9) {$5$};
	\node at (9.3,1.5) {$6$};
	\node at (9.3,2.1) {$7$};
	\node at (9.3,2.7) {$8$};
	\node at (9.3,3.3) {$9$};
	\node at (9.3,3.9) {$10$};
	\node at (9.3,4.5) {$11$};
	
	\draw[line width=1pt,dashed, green!40!black] (10,0) -- (10,5.2);
	
	\node at (10.8,0.8) {$a_0=0$};
	\node at (10.8,1.4) {$a_1=1$};
	\node at (10.8,2) {$a_2=3$};
	\node at (10.8,2.6) {$a_3=8$};
	\node at (10.8,3.2) {$a_4=15$};

	\node at (8.1,-1) {(b)};

		\end{tikzpicture}
		\caption{These figures depict machines with capacitated constant cost functions. Figure (a) shows three machines whose cost is $1$, $2$, and $15$ for any load up to $1$, $3$, and $6$, respectively (and the cost becomes infinite for any load beyond that). Similarly, Figure (b) shows five machines whose cost is $1,2,2,2$, and $7$ for any load up to $1,2,2,2$, and $8$, respectively. In both figures, the $a_k$ values are given on the right, and each number inside the machines represents a job with the number indicating the order of their arrival. The figures show how the \textsc{Delayed-OPT} algorithm assigns the jobs to the machines, e.g. the first job is assigned to the first machine in both cases.}
		\label{fig:onlineExample}
\end{figure}
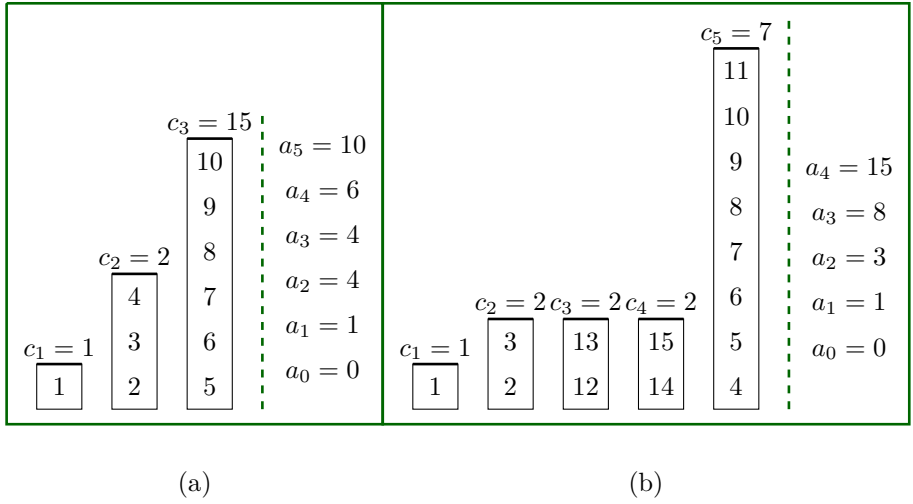

\begin{algorithm}
\DontPrintSemicolon
$q\leftarrow 0$ \tcp*[f]{Initialize counter for the number of jobs}\;
$\ell_j \leftarrow 0$ for each $j\in M$ \tcp*[f]{Initialize all loads to zero}\;
\While{there exist more jobs}
{
	$q\leftarrow q+1$\;
	$k\leftarrow \min \{k\in \mathbb{N}~|~ \exists j\in M: \ell_j < \ell_{jk}^*\}$\;
	$j\leftarrow \arg\min \{j\in M~|~ \ell_j < \ell_{jk}^*\}$\;
	$\ell_j \leftarrow \ell_j +1$ \tcp*[f]{Assign job to first machine that has not reached target load}
}
 \caption{\textsc{Delayed-OPT} Online Algorithm}\label{alg:DelayedOPT}
\end{algorithm}

\begin{lemma}
If $q\leq a_{k'}$ for some $k'\in \mathbb{N}$, then the value of $k$ computed by the algorithm in the iteration corresponding to the $q^\text{th}$ job satisfies $k\leq k'$. 
\end{lemma}
\begin{proof}
Assume that this is not the case. This would mean that in that iteration of the algorithm, for every machine $j\in M$ we have $\ell_j\geq \ell_{jk'}^*$. Summing over all $j\in M$, this would yield 
\begin{equation*}
\sum_{j\in M}\ell_j ~\geq~ \sum_{j\in M}\ell_{jk'}^* ~=~ a_{k'}.
\end{equation*}
But, since $q=\sum_{j\in M}\ell_j +1$, this contradicts the fact that $q\leq a_{k'}$.
\end{proof}

We now proceed to show that the competitive ratio of this algorithm is less than 4.

\begin{theorem}
The competitive ratio of the \textsc{Delayed-OPT} algorithm is less than 4.
\end{theorem}
\begin{proof}
Let $k\in \mathbb{N}$ be the minimum value such that $n\leq a_{k}$. The load that the algorithm assigns on any machine $j$ is no more than $\max_{k'\leq k}\{\load_{jk'}^*\}$. 
As a result, the cost of the \textsc{Delayed-OPT} algorithm for $n$ jobs is
\begin{equation*}
C(\online{n})\leq \sum_{k'\leq k} C(\opt(a_{k'})) < \sum_{k'\leq k} 2^{k'} < 2^{k+1},
\end{equation*}
while the optimal cost is $C(\opt(n))\geq C(\opt(a_{k-1}+1)) \geq 2^{k-1}$, leading to a competitive
ratio of less than $2^{k+1}/2^{k-1}=4$. 
\end{proof}

The following lemma will be useful in the next sections.

\begin{lemma}
\label{lem:fullSegmentFirst}
If $c_j(\ell)=c_j(\ell+\ell')$ for some machine $j$ and some loads $\ell,\ell'>0$, then right after the iteration that the \textsc{Delayed-OPT} algorithm assigns the $\ell^\text{th}$ job at machine $j$, it assigns the next $\ell'$ jobs at the same machine.
\end{lemma}

\begin{proof}
Suppose that in the iteration that the \textsc{Delayed-OPT} algorithm assigns the $\ell^\text{th}$ job at machine $j$ it computes $k$ to be the smallest value such that
there exists a machine $j'$ with $\ell_{j'}<\ell_{j'k}^*$. Since the algorithm assigns the current job to machine $j$, at this iteration $\ell_{j}<\ell_{jk}^*$ and $j$ has the smallest index among machines that satisfy this inequality. 

Moreover, since the cost functions are all non-decreasing, the cost of machine $j$ is the same for all loads between $\ell$ and $\ell+\ell'$, which means that $\ell_{jk}^* \geq \ell+\ell'$. To better see this suppose on the contrary that $\ell_{jk}^*<\ell+\ell'$. The allocation that assigns another job to machine $j$ has the same cost with current optimal allocation, i.e. $C(\opt(a_k))=C(\opt(a_k +1))$.  This is a contradiction to the definition of $a_k$ that needs to satisfy that $C(\opt(a_k))<C(\opt(a_k +1))$.

Overall, in the next iteration, $\ell_{j}<\ell_{jk}^*$ and $j$ should be the smallest index that satisfies this inequality, otherwise this wouldn't be true in the previous iteration. Therefore, the \textsc{Delayed-OPT} algorithm assigns the next job to machine $j$ and by induction it should assign all the following jobs until the load of machine $j$ becomes $\ell+\ell'$. Figure~\ref{fig:onlineExample} shows such examples.
\end{proof}

\section{Resource-Aware Mechanism for Games with Two Known Users}\label{sec:guar}

In this section we consider resource-aware mechanisms that are oblivious to the set of users in the system, with the exception of just two users. 
Formally, we consider classes of games such that for every game $G$ in this class, the set of agents, $S$, always contains two known agents. Note that the set $S$ is otherwise totally unrestricted and can also contain an adversarially chosen subset of the agents from ${\mathcal N}$, so this class of games is quite general. In fact, since the optimal allocation may very heavily depend on the total number of agents that participate in the game, the aforementioned restriction is seemingly benign. In what follows, we propose a resource-aware mechanism that assigns a special role to the two known agents, leading to very efficient equilibria for any \bounded\ cost function. In fact we show that the assignment of every Nash equilibrium in the induced game is the same as the outcome of the \textsc{Delayed-OPT} algorithm when scheduling $|S|$ jobs.

As a warm-up, we first consider the games whose cost functions are drawn from the class of capacitated cost functions, and then we go on to extend our result beyond this class.

\subsection{Warm-up: A Class of Capacitated Constant Functions}
\label{sec:capacitatedGuaranteedDisruptors}
In order to more clearly capture the intuition behind how our proposed mechanism works, we first focus on games whose cost functions are capacitated constant, with a capacity of at least 4. That is, for every machine $j$ we have $c_j(\ell)=c_j$ when $\ell\leq \cp_j$ and $c_j(\ell)=\infty$ otherwise, where $c_j > 0$ and $\cp_j\geq 4$ are constants\footnote{\label{foot:positiveCosts}The assumption of $c_j>0$ for all $j$ is w.l.o.g.\ because if there are machines with zero cost, we may charge everybody with $0$, unless the machine load exceeds its capacity, in which case everybody is charged with infinity. In both the \textsc{Delayed-OPT} algorithm and any Nash equilibrium those machines are firstly occupied up to their capacity and then other machines are used resulting in a PoA equal to the one that ignores those machines.}. Note that these cost functions are actually not bounded, since they jump from some constant to infinity when their capacity is exceeded, so this section also shows that our positive results can even be extended to cost functions beyond the class of bounded ones.

Before presenting our protocol, we make an important observation, that can be derived directly from Lemma~\ref{lem:fullSegmentFirst}, regarding the allocation of the \textsc{Delayed-OPT} algorithm when the machines have capacitated constant cost functions. 

\begin{observation}
\label{obs:machineOrder}
For any instance involving a set $M$ of machines with capacitated constant cost functions, there exists an ordering of the machines in $M$ such that the \textsc{Delayed-OPT} algorithm fills up machine $j$ up to its capacity before assigning any job to any machine $j'$ that is later in the ordering.
\end{observation}

Let $D=\{1,2\}$ be the set of the two agents, called {\em enforcers}, who are guaranteed to participate, and let $R=S\setminus D$ be the rest of the agents, which we call {\em regular} agents. Also, given some set of agents $S'$, let $\high(S')$ be the first (highest priority) agent in $S'$ according to a global ordering $\pi$. For simplicity we assume that the machines are renamed according to the ordering implied by the \textsc{Delayed-OPT} algorithm (Observation~\ref{obs:machineOrder}), and let $\CumCost_j=\sum_{k\leq j}c_k$ be the sum of the costs of the first $j$ machines in this ordering. Note that the value of $\CumCost_j$ is {\em strictly} increasing with $j$ by our convention that $c_j>0$ for all $j$. 
Finally, to define the protocol we also use an arbitrarily small positive value $\ve_j$ for each machine $j$ to be used as a special charge for \diss\ in some cases; $\ve_j$ values are strictly decreasing values, i.e. $\ve_j > \ve_{j+1}$.

{\bf Brief description of the protocol.} The \diss\ are charged with the small value $\ve_j$ for using machine $j$ only in two cases: i) if they are together in $j$ along with at least one regular agent (if there were no regular agent, the \diss\ should cover the cost of the machine) and the load of machine $j$ doesn't exceed its capacity $\cp_j$, ii) if the \dis\ is alone in $j$ and the load of machine $j$ exceeds its capacity $\cp_j$. In any other case they pay $\CumCost_j$. Regarding the regular agents, the highest priority regular agent always pays a non-zero charge. More specifically, if machine $j$'s load doesn't exceed $\cp_j$, the highest priority regular agent pays the cost of $j$, $c_j$, if the machine is full (i.e. its load equals $\cp_j$) and there is no \dis\ in $j$; otherwise, meaning when $j$'s load is less than $\cp_j$ or there is an \dis\ in $j$, the highest priority regular agent pays $\CumCost_j$. The rest of the regular agents are charged with $0$ if $j$'s load doesn't exceed $\cp_j$. If $j$'s load exceeds $\cp_j$, then everybody is charged with infinity.

\paragraph{Protocol.} 
Given a strategy profile $\prof$, the cost share of any \dis\ $i \in \disSet$ for using machine $j$ is
\[
	\xi_i (\prof) =
	\begin{cases}
	\ve_j & \text{if } \load_j(\prof) \leq \cp_j \mbox{ and } \disSet \subset S_j(\prof)\\
	\ve_j & \text{if } \load_j(\prof) > \cp_j \mbox{ and } \disSet \cap S_j(\prof) = \{i\}\\
	\CumCost_j & \text{otherwise.}
	\end{cases}\,
	\]
	
The cost share of any regular agent $i \in R$ for using machine $j$ is
	\[
	\xi_i (\prof) =
	\begin{cases}
	0 & \text{if }   \load_j(\prof) \leq \cp_j \mbox{ and } i \neq \high(S_j(\prof)\cap R)\\
	c_j & \text{if } \load_j(\prof) = \cp_j, \disSet \cap S_j(\prof)= \emptyset  \mbox{ and }  i = \high(S_j(\prof)\cap R)\\
	\CumCost_j & \text{if } \load_j(\prof) = \cp_j, \disSet \cap S_j(\prof)\neq \emptyset \mbox{ and }  i = \high(S_j(\prof)\cap R)\\
	\CumCost_j & \text{if } \load_j(\prof) < \cp_j \mbox{ and }  i = \high(S_j(\prof)\cap R)\\
	\infty & \text{otherwise.}
	\end{cases}\,
	\]

The main idea behind this protocol is that in the equilibrium if agents use some machine, all machines with lower indices should be full (i.e. its load equals its capacity). As we mentioned above, $\CumCost_j$ values are strictly increasing. 
As a result if some agent is charged $\CumCost_j$ in machine $j$, she prefers to deviate to a non-full machine (non-full means that its load is less that its capacity) with smaller index. Such an agent exists when the machine is not full or when an \dis\ is using it. However, there is no such agent when the machine is full with only regular agents, where the importance of \diss\ comes in place as we explain next. We note here that it is crucial to keep the budget balance in full machines without \diss\ so that we do not lose in efficiency too much.

If a machine that is not used by the \textsc{Delayed-OPT} algorithm is full with only regular agents, \diss\ are going to disrupt them and push them to machines with lower indices. The reason is because $\ve_j$ values are decreasing, so  \diss\ prefer to occupy machines with higher indices. So, if an \dis\ deviates to a full machine $j$, the load of that machine will exceed capacity and the \dis\ will be charged $\ve_j$. 

The cases where the charges of the \diss\ are high ($\CumCost_j$) are crucial in order to guarantee stability as we show in Theorem~\ref{lem:disGuaranteedStability}.

\begin{theorem} The PoA for the class of capacitated constant cost functions, assuming two \diss, is constant. 
\label{lem:disGuaranteedPoA}
\end{theorem}

 \begin{proof}
It is sufficient to show that the social cost of any Nash equilibrium is constant away from the cost induced by the \textsc{Delayed-OPT} algorithm, which in turn is constant away from the cost of the optimal allocation. 

In fact we show that under any pure Nash equilibrium, the allocation is the same with the outcome of the \textsc{Delayed-OPT} algorithm; that is for any used machine $r$, all prior machines $j<r$ are fully used. Then, it is easy to check that, regarding the overcharging, each \dis\ may "cause" some regular agent to pay at most the cost of the outcome of the \textsc{Delayed-OPT} algorithm and each \dis\ itself may pay some arbitrarily small value $\ve_j$.\footnote{There is no Nash equilibrium where the \diss\ are charged more than some $\ve_j$, unless there is no regular agent where we again have the same overcharging.}

For the sake of contradiction suppose that in some Nash equilibrium there exist machines $j<r$ such that machine $r$ is used and machine $j$ is not full. Let $r$ be the largest possible such index.
\begin{itemize}
\item If $r$ is not full, or if it is full and has at least one \dis , there exists an agent paying $\CumCost_r$ and if he deviates to $j$ he should pay at most $\CumCost_j<\CumCost_r$, so he has an incentive to deviate (Figure~\ref{fig:disGuaranteedPoA1}). 
\item If $r$ is full with only regular agents, there exists an \dis\ in an earlier machine $j'<r$ paying at least $\ve_{j'}$. 
That \dis\ has an incentive to deviate to $r$ where he will pay $\ve_r < \ve_{j'}$ (Figure~\ref{fig:disGuaranteedPoA2}).  
\end{itemize}

\begin{figure}[htb]
	\centering
	\begin{tikzpicture}
	

	\draw[line width=1pt, green!40!black] (-1.25,3.2) -- (6.25,3.2) -- (6.25,-0.4) -- (-1.25, -0.4) -- (-1.25,3.2);
	\draw[line width=1pt, green!40!black] (6.25,3.2) -- (11.75,3.2) -- (11.75,-0.4) -- (6.25, -0.4) -- (6.25,3.2);

	\draw (-1,0) rectangle (-0.5,2);
	\draw[line width=1pt] (-1,2) -- (-0.5,2);
	\node at (-0.75,2.2) {$c_{1}$};
	\node at (-0.75,-0.2) {$1$};
	\draw[fill = blue] (-0.75,0.2) circle (3pt);
	\draw[fill = blue] (-0.75,0.6) circle (3pt);
	\draw[fill = blue] (-0.75,1) circle (3pt);
	\draw[fill = blue] (-0.75,1.4) circle (3pt);
	\draw[fill = blue] (-0.75,1.8) circle (3pt);
		
	\draw[fill = black] (0,-0.2) circle (1pt);
	\draw[fill = black] (0.25,-0.2) circle (1pt);
	\draw[fill = black] (0.5,-0.2) circle (1pt);

	\draw (1,0) rectangle (1.5,2.8);
	\draw[line width=1pt] (1,2.8) -- (1.5,2.8);
	\node at (1.25,3) {$c_{j}$};
	\node at (1.25,-0.2) {$j$};
	\draw[fill = blue] (1.25,0.2) circle (3pt);
	\draw[fill = blue] (1.25,0.6) circle (3pt);
	\draw[fill = blue] (1.25,1) circle (3pt);
		
	\draw[fill = black] (2,-0.2) circle (1pt);
	\draw[fill = black] (2.25,-0.2) circle (1pt);
	\draw[fill = black] (2.5,-0.2) circle (1pt);

	\draw (3,0) rectangle (3.5,2.4);
	\draw[line width=1pt] (3,2.4) -- (3.5,2.4);
	\node at (3.25,2.6) {$c_{r}$};
	\node at (3.25,-0.2) {$r$};
	\draw[fill = blue] (3.25,0.2) circle (3pt);
	\draw[fill = blue] (3.25,0.6) circle (3pt);
	\draw[fill = blue] (3.25,1) circle (3pt);
	\draw[fill = blue] (3.25,1.4) circle (3pt);
	\node at (5,0.35) {{\small highest priority}};
	\node at (5,0.05) {{\small regular agent}};
	\draw[line width=1pt,<-, blue] (3.5,0.2) -- (4,0.2);
	\draw[line width=1pt,<-, dashed, blue] (1.5,1.4) -- (3,0.2);	
	\node [rotate=-38] at (2.35,1) {deviates};

	\node at (2.5,-1) {(a)};

	\draw (6.5,0) rectangle (7,2.8);
	\draw[line width=1pt] (6.5,2.8) -- (7,2.8);
	\node at (6.75,3) {$c_{j}$};
	\node at (6.75,-0.2) {$j$};
	\draw[fill = blue] (6.75,0.2) circle (3pt);
	\draw[fill = blue] (6.75,0.6) circle (3pt);
	\draw[fill = blue] (6.75,1) circle (3pt);
		
	\draw[fill = black] (7.5,-0.2) circle (1pt);
	\draw[fill = black] (7.75,-0.2) circle (1pt);
	\draw[fill = black] (8,-0.2) circle (1pt);

	\draw (8.5,0) rectangle (9,2.4);
	\draw[line width=1pt] (8.5,2.4) -- (9,2.4);
	\node at (8.75,2.6) {$c_{r}$};
	\node at (8.75,-0.2) {$r$};
	\draw[fill = blue] (8.75,0.2) circle (3pt);
	\draw[fill = blue] (8.75,0.6) circle (3pt);
	\draw[fill = red] (8.75,1) circle (3pt);
	\draw[fill = blue] (8.75,1.4) circle (3pt);
	\draw[fill = blue] (8.75,1.8) circle (3pt);
	\draw[fill = blue] (8.75,2.2) circle (3pt);
	\node at (10.5,0.35) {{\small highest priority}};
	\node at (10.5,0.05) {{\small regular agent}};
	\draw[line width=1pt,<-, blue] (9,0.2) -- (9.5,0.2);
	\draw[line width=1pt,<-, dashed, blue] (7,1.4) -- (8.5,0.2);	
	\node [rotate=-38] at (7.85,1) {deviates};
	\node at (10.2,1) {{\small \dis}};
	\draw[line width=1pt,<-, red] (9,1) -- (9.5,1);
	\node at (9,-1) {(b)};

		\end{tikzpicture}
		\caption{In this figure we assume that machine $j$ is not full and there exists a non empty machine $r$, with $r>j$ (where $r$ is the maximum such index). If machine $r$ is either not full (a) or has an \dis\ (b), then there is always a regular agent from $r$ that prefers to deviate to $j$.}
		\label{fig:disGuaranteedPoA1}
\end{figure}
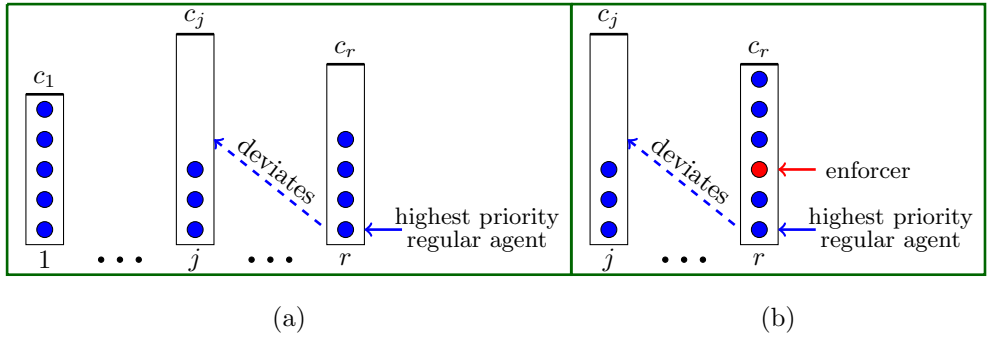

\begin{figure}[htb]
	\centering
	\begin{tikzpicture}

	\draw (-2,0) rectangle (-1.5,2);
	\draw[line width=1pt] (-2,2) -- (-1.5,2);
	\node at (-1.75,2.2) {$c_{1}$};
	\node at (-1.75,-0.2) {$1$};
	\draw[fill = blue] (-1.75,0.2) circle (3pt);
	\draw[fill = blue] (-1.75,0.6) circle (3pt);
	\draw[fill = blue] (-1.75,1) circle (3pt);
	\draw[fill = blue] (-1.75,1.4) circle (3pt);
	\draw[fill = blue] (-1.75,1.8) circle (3pt);
		
	\draw[fill = black] (-1,-0.2) circle (1pt);
	\draw[fill = black] (-0.75,-0.2) circle (1pt);
	\draw[fill = black] (-0.5,-0.2) circle (1pt);
	
	\draw (0,0) rectangle (0.5,2.8);
	\draw[line width=1pt] (0,2.8) -- (0.5,2.8);
	\node at (0.25,3) {$c_{j}$};
	\node at (0.25,-0.2) {$j$};
	\draw[fill = blue] (0.25,0.2) circle (3pt);
	\draw[fill = blue] (0.25,0.6) circle (3pt);
	\draw[fill = blue] (0.25,1) circle (3pt);
		
	\draw[fill = black] (1.5,-0.2) circle (1pt);
	\draw[fill = black] (1.75,-0.2) circle (1pt);
	\draw[fill = black] (2,-0.2) circle (1pt);

	\draw (3,0) rectangle (3.5,2);
	\draw[line width=1pt] (3,2) -- (3.5,2);
	\node at (3.25,2.2) {$c_{j'}$};
	\node at (3.25,-0.2) {$j'$};
	\draw[fill = blue] (3.25,0.2) circle (3pt);
	\draw[fill = blue] (3.25,0.6) circle (3pt);
	\draw[fill = blue] (3.25,1) circle (3pt);
	\draw[fill = red] (3.25,1.4) circle (3pt);
	\draw[fill = blue] (3.25,1.8) circle (3pt);
	
	\node at (1.8,1.4) {{\small \dis}};
	\draw[line width=1pt,->, red] (2.5,1.4) -- (3,1.4);

	\draw[fill = black] (4,-0.2) circle (1pt);
	\draw[fill = black] (4.25,-0.2) circle (1pt);
	\draw[fill = black] (4.5,-0.2) circle (1pt);

	\draw (5,0) rectangle (5.5,2.4);
	\draw[line width=1pt] (5,2.4) -- (5.5,2.4);
	\node at (5.25,2.6) {$c_{r}$};
	\node at (5.25,-0.2) {$r$};
	\draw[fill = blue] (5.25,0.2) circle (3pt);
	\draw[fill = blue] (5.25,0.6) circle (3pt);
	\draw[fill = blue] (5.25,1) circle (3pt);
	\draw[fill = blue] (5.25,1.4) circle (3pt);
	\draw[fill = blue] (5.25,1.8) circle (3pt);
	\draw[fill = blue] (5.25,2.2) circle (3pt);
	\draw[line width=1pt,->, dashed, red] (3.5,1.4) -- (5,2.4);	
	\node [rotate=33] at (4.15,2.1) {deviates};

		\end{tikzpicture}
		\caption{In this figure we assume that machine $j$ is not full and there exists a non empty machine $r$, with $r>j$ (where $r$ is the maximum such index). If machine $r$ is full with only regular agents, then any \dis\ prefers to deviate to $r$.}
		\label{fig:disGuaranteedPoA2}
\end{figure}
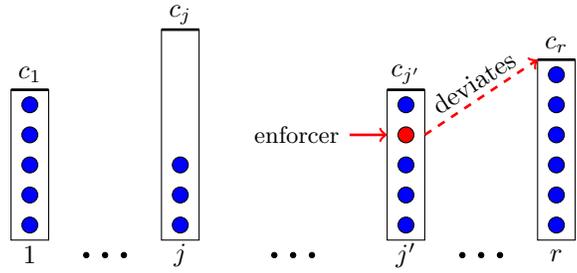

In both cases there exists an agent with an incentive to deviate to another machine which is a contradiction to our assumption that this is a Nash equilibrium.
\end{proof}

\begin{theorem} 
The protocol for the class of capacitated constant cost functions, assuming two \diss, is stable.
\label{lem:disGuaranteedStability}
\end{theorem}
\begin{proof}
In order to show stability, we create a strategy profile that is an equilibrium for any set of agents $S\subseteq \mathcal{N}$ as long as $D\subseteq S$.

Let $n$ be the number of agents in the system, where $n-2$ of them are regular agents, since there exist two \diss. Suppose that $r$ machines are occupied based on the \textsc{Delayed-OPT} algorithm, with the first $r-1$ machines being fully occupied and machine $r$ having $n_r\leq \cp_r$ agents. The strategy profile we create depends on the value of $n_r$.

\paragraph{Case of $n_r\leq 2$.} In this case, we create a strategy profile where the \diss\ use the last full machine $r-1$ (unless $r=1$, meaning that there is no regular agent, and the \diss\ use machine $1$ which is a Nash equilibrium). The regular agents are placed according to the outcome of the \textsc{Delayed-OPT} algorithm such that in the last machine $r$ the {\em lowest} priority agents are placed (Figure~\ref{fig:disGuaranteedStability} (a)). 
Next we show that nobody has an incentive to deviate from this strategy profile and therefore it is stable. 

The \diss\ are currently charged with $\ve_{r-1}$ and if they deviate to any previous machine $j$ with $j<r-1$ they will be charged with $\ve_j>\ve_{r-1}$. Moreover, if they unilaterally deviate to $r$ they will be charged $\CumCost_r > \ve_{r-1}$ because they will be the only \dis\ there. Deviating to any other machine $j$ with $j>r$ will result in an even higher charge, since the \dis\ will be alone there. Overall, \diss\ have no incentive to deviate.

From the regular agents' perspective, nobody has an incentive to deviate to a full machine $j'<r$ because its load then will exceed its capacity resulting in infinity charges. Additionally, no agent currently located to some machine $j\leq r$ has an incentive to deviate to an empty machine $j'>r$, because he is currently charged at most $\CumCost_j$ and if he deviates to $j'$, he will be charged $\CumCost_{j'} > \CumCost_j$. The last case to check is if an agent currently located to some machine $j<r$, has an incentive to deviate to $r$.  
Note that if he deviates to $r$, the machine will still not be full and he will be the highest priority agent, as in $r$ we allocated the lowest priority agents; therefore, he will be charged $\CumCost_r > \CumCost_j$, where $\CumCost_j$ is the maximum he may currently be charged. 

\paragraph{Case of $n_r> 2$.}
In this case, we create a strategy profile where the \diss\ use that last machine $r$. The regular agents are placed according to the outcome of the \textsc{Delayed-OPT} algorithm such that in the last machine $r$ the {\em lowest} priority agents are placed (Figure~\ref{fig:disGuaranteedStability} (b)). 
Similar arguments hold in this case in order to show that nobody has an incentive to deviate from this strategy profile. 

More specifically the \diss\ are currently charged with $\ve_{r}$ and any deviation will result in a charge of either $\ve_j$ with $j<r$ or $\CumCost_j$ with $j>r$, which are both strictly greater than $\ve_r$. 

Regarding the regular agents, as before, nobody wants to deviate to a full machine or to a machine $j>r$. Any agent currently using some machine $j<r$ is charged with at most $\CumCost_j$ and if he deviated to machine $r$ he would pay at least $\CumCost_{r}> \CumCost_j$ because he would be the highest priority agent in $r$.
\end{proof}

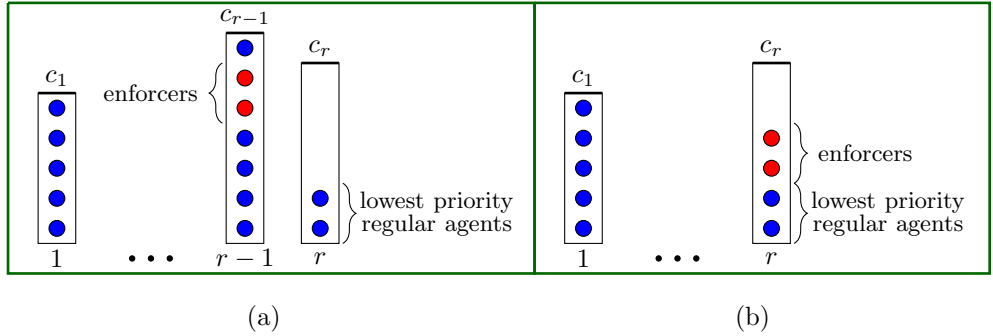
\begin{figure}[htb]
	\centering
	\begin{tikzpicture}
	
	\draw[line width=1pt, green!40!black] (-1.9,3.2) -- (5.1,3.2) -- (5.1,-0.4) -- (-1.9, -0.4) -- (-1.9,3.2);
	\draw[line width=1pt, green!40!black] (5.1,3.2) -- (11.15,3.2) -- (11.15,-0.4) -- (5.1, -0.4) -- (5.1,3.2);
	
	\draw (-1.5,0) rectangle (-1,2);
	\draw[line width=1pt] (-1.5,2) -- (-1,2);
	\node at (-1.25,2.2) {$c_{1}$};
	\node at (-1.25,-0.2) {$1$};
	\draw[fill = blue] (-1.25,0.2) circle (3pt);
	\draw[fill = blue] (-1.25,0.6) circle (3pt);
	\draw[fill = blue] (-1.25,1) circle (3pt);
	\draw[fill = blue] (-1.25,1.4) circle (3pt);
	\draw[fill = blue] (-1.25,1.8) circle (3pt);
		
	\draw[fill = black] (-0.25,-0.2) circle (1pt);
	\draw[fill = black] (0,-0.2) circle (1pt);
	\draw[fill = black] (0.25,-0.2) circle (1pt);

	\draw (1,0) rectangle (1.5,2.8);
	\draw[line width=1pt] (1,2.8) -- (1.5,2.8);
	\node at (1.25,3) {$c_{r-1}$};
	\node at (1.25,-0.2) {$r-1$};
	\draw[fill = blue] (1.25,0.2) circle (3pt);
	\draw[fill = blue] (1.25,0.6) circle (3pt);
	\draw[fill = blue] (1.25,1) circle (3pt);
	\draw[fill = blue] (1.25,1.4) circle (3pt);
	\draw[fill = red] (1.25,1.8) circle (3pt);
	\draw[fill = red] (1.25,2.2) circle (3pt);
	\draw[fill = blue] (1.25,2.6) circle (3pt);
	\draw [decorate,decoration={brace,amplitude=5pt}] (0.95,1.6) -- (0.95,2.4) node [black,midway,xshift=-0.5cm] {};
	\node at (0,2) {{\small \diss}};

	\draw (2,0) rectangle (2.5,2.4);
	\draw[line width=1pt] (2,2.4) -- (2.5,2.4);
	\node at (2.25,2.6) {$c_{r}$};
	\node at (2.25,-0.2) {$r$};
	\draw[fill = blue] (2.25,0.2) circle (3pt);
	\draw[fill = blue] (2.25,0.6) circle (3pt);
	\draw [decorate,decoration={brace,mirror,amplitude=5pt}] (2.55,0) -- (2.55,0.8) node [black,midway,xshift=-0.5cm] {};
	\node at (3.8,0.55) {{\small lowest priority}};
	\node at (3.8,0.25) {{\small regular agents}};

	\node at (1.5,-1) {(a)};

	\draw (5.5,0) rectangle (6,2);
	\draw[line width=1pt] (5.5,2) -- (6,2);
	\node at (5.75,2.2) {$c_{1}$};
	\node at (5.75,-0.2) {$1$};
	\draw[fill = blue] (5.75,0.2) circle (3pt);
	\draw[fill = blue] (5.75,0.6) circle (3pt);
	\draw[fill = blue] (5.75,1) circle (3pt);
	\draw[fill = blue] (5.75,1.4) circle (3pt);
	\draw[fill = blue] (5.75,1.8) circle (3pt);
		
	\draw[fill = black] (6.75,-0.2) circle (1pt);
	\draw[fill = black] (7,-0.2) circle (1pt);
	\draw[fill = black] (7.25,-0.2) circle (1pt);
	
	\draw (8,0) rectangle (8.5,2.4);
	\draw[line width=1pt] (8,2.4) -- (8.5,2.4);
	\node at (8.25,2.6) {$c_{r}$};
	\node at (8.25,-0.2) {$r$};
	\draw[fill = blue] (8.25,0.2) circle (3pt);
	\draw[fill = blue] (8.25,0.6) circle (3pt);
	\draw[fill = red] (8.25,1) circle (3pt);
	\draw[fill = red] (8.25,1.4) circle (3pt);
	\draw [decorate,decoration={brace,mirror,amplitude=5pt}] (8.55,0) -- (8.55,0.8) node [black,midway,xshift=-0.5cm] {};
	\node at (9.8,0.55) {{\small lowest priority}};
	\node at (9.8,0.25) {{\small regular agents}};
	\draw [decorate,decoration={brace,mirror,amplitude=5pt}] (8.55,0.8) -- (8.55,1.6) node [black,midway,xshift=-0.5cm] {};
	\node at (9.5,1.2) {{\small \diss}};
	
	\node at (8,-1) {(b)};

		\end{tikzpicture}
		\caption{This figure shows the stable outcomes when the \textsc{Delayed-OPT} algorithm allocates in the last machine (a) at most two agents and (b) more than two agents.}
		\label{fig:disGuaranteedStability}
\end{figure}

In Appendix~\ref{sec:appendixdisc} we give some intuition on why we may need of at least two \diss\ and the capacities to be at least $4$. Both restrictions are important in order to guarantee stability.

\subsection{Bounded Cost Functions}
\label{sec:genGuaranteedDisruptors}

We now extend the result of Section~\ref{sec:capacitatedGuaranteedDisruptors} to the class of bounded cost functions. For simplicity, we focus on the class of 4-step cost functions which naturally generalize the capacitated cost functions considered above; as we discussed in Section~\ref{sec:prelim}, any bounded cost function can be approximated by a 4-step cost function, so our results directly extend to bounded cost functions as well. 

A key difference between the segments of 4-step functions and the capacitated constant functions is that having a single job in a segment may have two meanings in the respective machines with capacitated constant functions: it may be considered as i) having a single job in the machine with capacitated constant function corresponding to that segment or ii) having an overload in the machine with capacitated constant function corresponding to the previous segment of the step function. In order to overcome this ambiguity we slightly change our protocol in order to handle those two cases consistently and get the same results. 

Let agents $\disSet=\{1,2\}$ be the two agents/\diss\ who are guaranteed to participate, and let $R=S\setminus\disSet$ be the {\em regular} agents. Before describing the protocol we need to give some further definitions; Figure~\ref{fig:GendisGuaranteed} gives some intuition for some of the following definitions.

\subsubsection{Preliminaries}\hspace{5pt}

We first provide an alternative definition of a 4-step function. In the rest of the paper, we will be assuming that all the machine cost functions are 4-step functions.
\begin{definition}
\label{def:stepFunction}
A function $c$ is called {\em 4-step function} if the following are true: there are steps of lengths $\cp(1),\cp(2),\dots\geq 4$ 
such that for all $k$ and all $x\in [ 1+\sum_{k'=1}^{k-1} \cp({k'}) ,\sum_{k'=1}^{k}  \cp({k'})]$ we have that

 $$c(x)= c\left(\sum^{k}_{k'=1} \cp({k'})\right)=\tilde{c}(k).$$

\end{definition}
 That is, the cost function increases only when an extra step needs to be used. If any number of jobs between one and $\cp(1)$ are undertaken by this machine, the cost is $\tilde{c}(1)$. Then if one more job is added the cost jumps to $\tilde{c}(2)$ and then the cost for $\cp(1)+1$ up to $\cp(1)+\cp(2)$ jobs remains $\tilde{c}(2)$, and so forth. Note that trivially all functions on natural numbers are 1-step functions.

{\bf Length and cost of a segment.} According to Definition~\ref{def:stepFunction}, we define segment $k$ of machine $j$ to be the $k^{th}$ step of machine $j$'s cost function $c_j$ and has length $\cp_j(k)$ and cost $\tilde{c}_j(k)$. 

{\bf Last used segment $\lastSeg_j(\ell_j(\prof)) = \lastSeg_j(\prof)$.} For each machine $j$ and profile $\prof$, we denote by $\lastSeg_j(\prof)$ the last segment that is used in machine $j$ under $\prof$. It holds that $c_j(\prof) = \tilde{c}_j(\lastSeg_j(\prof))$.

{\bf Machine's excess $w_j(\ell_j(\prof))=w_j(\prof)$.} For each machine $j$ and profile $\prof$, we denote by $w_j(\prof)$ the number of jobs occupying the last segment of machine $j$ under $\prof$ if that segment is not filled to capacity. If the number of jobs fill the last segment to capacity then we set  $w_j(\prof)$ to $0$. More formally,

\[ 
	w_j (\prof) =
	\begin{cases}
	 \ell_j(\prof) - \sum^{\lastSeg_j(\prof)-1}_{k=1}\cp_j(k)   & \text{ if }  \ell_j(\prof) > \sum^{\lastSeg_j(\prof)-1}_{k=1}\cp_j(k) \\
0 & \text{ otherwise}
	\end{cases}\,
	\]

{\bf Segment order $\phi$.} Lemma~\ref{lem:fullSegmentFirst} implies that the \textsc{Delayed-OPT} algorithm fills up a segment up to its capacity before assigning any job to any other segment.
 Therefore, a priority order on the segments can be derived according to the \textsc{Delayed-OPT} algorithm that assigns for every machine $j$ and segment $k$ a number $\phi_{j}(k)$. The function $\phi$ respects the order of the machine step costs, i.e. it is strictly monotone.

{\bf Definition of $\CumCost_j(k)$.} Similarly to the case of capacitated constant functions, we define $$\CumCost_j(k)=\sum_{j'} \max_{k': \phi_{j'}(k') \leq \phi_{j}(k)} \tilde{c}_{j'}(k'),$$
which is the aggregate cost of all machines if all segments up to $\phi_j(k)$ in the priority order are occupied. 
W.l.o.g. the $\CumCost_j(k)$ values are {\em strictly} increasing according to the order defined by $\phi$. This is by assuming non-zero costs which is w.l.o.g. according to footnote~\ref{foot:positiveCosts}, and additionally if two consecutive segments have the same cost, we may assume that they are merged into a single segment.

{\bf Definition of $\ve_j(k)$.} We also use an arbitrarily small positive value $\ve_j(k)$ for each segment $k$ of machine $j$ to be used as a special charge for \diss\ in some cases; $\ve_j(k)$ values are strictly decreasing values according to the order $\phi$, i.e. if $\phi_j(k) > \phi_{j'}(k')$ then $\ve_j(k) < \ve_{j'}(k')$.

{\bf First two machines.} We further distinguish two machines $1$ and $2$ to be the first and the second machines, respectively, to be used by the \textsc{Delayed-OPT} algorithm.

{\bf Highest priority agents $\high(S')$.} Given some set of agents $S'$, $\high_i(S')$ is the $i^{th}$ agent in $S'$ according to the global ordering $\pi$. 

\paragraph{Protocol.} 
 Given a strategy profile $\prof$ we next define the cost shares of the agents. For simplicity, we drop the dependency on the load and on $\prof$ since there is no ambiguity. The cost share of any \dis\ $i \in \disSet$ using machine $j$ is
\[
	\xi_i (\prof) =
	\begin{cases}
	\ve_{j}(\lastSeg_j) & \text{if } w_j \neq 1  \mbox{ and }  \disSet \subset S_j \\
	\ve_{j}(\lastSeg_j-1) & \text{if } w_j = 1, \disSet \cap S_j = \{i\} \mbox{ and } \lastSeg_j >1\\
	\CumCost_{j}(\lastSeg_j) & \text{otherwise.}
	\end{cases}\,
	\]
	
	The cost share of any regular agent $i \in R$ using machine $j$ is 
	\[
	\xi_i (\prof) =
	\begin{cases}
	c_j & \text{if } w_j=0, \disSet \cap S_j =\emptyset \mbox{ and } i = \high_1(S_j\cap R) \\
	\CumCost_{j}(\lastSeg_j) & \text{if } w_j=0, \disSet \cap S_j \neq \emptyset \mbox{ and } i = \high_1(S_j\cap R)\\
	\CumCost_{j}(\lastSeg_j) & \text{if } w_j=1 \mbox{ and } i \in \{\high_1(S_j\cap R), \high_2(S_j\cap R)\} \\
	\CumCost_{j}(\lastSeg_j) & \text{if } w_j\notin \{0,1\}  \mbox{ and } i = \high_1(S_j\cap R) \\
	0 & \text{otherwise.}
	\end{cases}\,
	\]

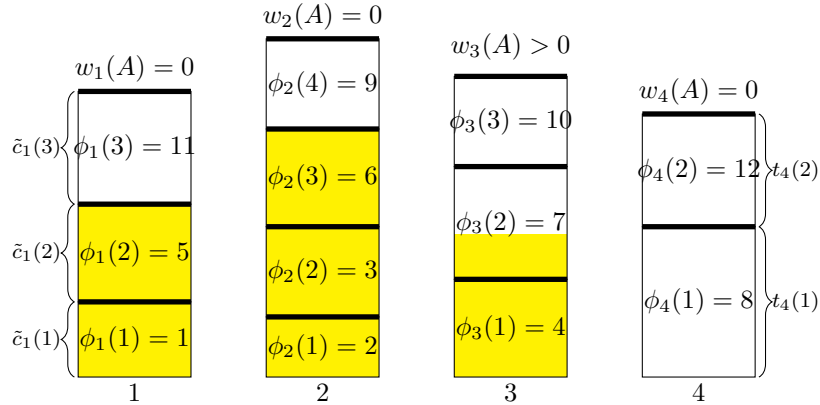
\begin{figure}[htb]
	\centering
	\begin{tikzpicture}
	\draw [fill = yellow] (0.25,0) rectangle (1.75,1);
	\draw [fill = yellow] (0.25,1) rectangle (1.75,2.3);
	\draw (0.25,2.3) rectangle (1.75,3.8);
	\node at (1,4.1) {$w_{1}(A)=0$};
	\draw[line width=2pt] (0.25,1) -- (1.75,1);
	\draw [decorate,decoration={brace,amplitude=5pt}] (0.2,0) -- (0.2,1) node [black,midway,xshift=-0.5cm] {\footnotesize $\tilde{c}_1(1)$};
	\node at (1,0.5) {$\phi_{1}(1)=1$};
	\draw[line width=2pt] (0.25,2.3) -- (1.75,2.3);
	\draw [decorate,decoration={brace,amplitude=5pt}] (0.2,1) -- (0.2,2.3) node [black,midway,xshift=-0.5cm] {\footnotesize $\tilde{c}_1(2)$};
	\node at (1,1.65) {$\phi_{1}(2)=5$};
	\draw[line width=2pt] (0.25,3.8) -- (1.75,3.8);
	\draw [decorate,decoration={brace,amplitude=5pt}] (0.2,2.3) -- (0.2,3.8) node [black,midway,xshift=-0.5cm] {\footnotesize $\tilde{c}_1(3)$};
	\node at (1,3.05) {$\phi_{1}(3)=11$};
	
	\draw [fill = yellow] (2.75,0) rectangle (4.25,0.8);
	\draw [fill = yellow] (2.75,0.8) rectangle (4.25,2);
	\draw [fill = yellow] (2.75,2) rectangle (4.25,3.3);
	\draw (2.75,3.3) rectangle (4.25,4.5);
	\node at (3.5,4.8) {$w_{2}(A)=0$};
	\draw[line width=2pt] (2.75,0.8) -- (4.25,0.8);
	\node at (3.5,0.4) {$\phi_{2}(1)=2$};
	\draw[line width=2pt] (2.75,2) -- (4.25,2);
	\node at (3.5,1.4) {$\phi_{2}(2)=3$};
	\draw[line width=2pt] (2.75,3.3) -- (4.25,3.3);
	\node at (3.5,2.65) {$\phi_{2}(3)=6$};
	\draw[line width=2pt] (2.75,4.5) -- (4.25,4.5);
	\node at (3.5,3.9) {$\phi_{2}(4)=9$};
	
	\draw [fill = yellow] (5.25,0) rectangle (6.75,1.3);
	\fill [yellow] (5.25,1.3) rectangle (6.75,1.9);
	\draw (5.25,1.3) rectangle (6.75,2.8);
	\draw (5.25,2.8) rectangle (6.75,4);
	\node at (6,4.4) {$w_{3}(A)>0$};
	\draw[line width=2pt] (5.25,1.3) -- (6.75,1.3);
	\node at (6,0.65) {$\phi_{3}(1)=4$};
	\draw[line width=2pt] (5.25,2.8) -- (6.75,2.8);
	\node at (6,2.05) {$\phi_{3}(2)=7$};
	\draw[line width=2pt] (5.25,4) -- (6.75,4);
	\node at (6,3.4) {$\phi_{3}(3)=10$};
	
	\draw (7.75,0) rectangle (9.25,2);
	\draw (7.75,2) rectangle (9.25,3.5);
	\node at (8.5,3.8) {$w_{4}(A)=0$};
	\draw[line width=2pt] (7.75,2) -- (9.25,2);
	\node at (8.5,1) {$\phi_{4}(1)=8$};
	\draw[line width=2pt] (7.75,3.5) -- (9.25,3.5);
	\node at (8.5,2.75) {$\phi_{4}(2)=12$};
	\draw [decorate,decoration={brace,mirror,amplitude=5pt}] (9.3,0) -- (9.3,2) node [black,midway,xshift=14pt] {\footnotesize $\cp_4(1)$};
	\draw [decorate,decoration={brace,mirror,amplitude=5pt}] (9.3,2) -- (9.3,3.5) node [black,midway,xshift=14pt] {\footnotesize $\cp_4(2)$};

	\node at (1,-0.2) {$1$};
	\node at (3.5,-0.2) {$2$};
	\node at (6,-0.2) {$3$};
	\node at (8.5,-0.2) {$4$};
	\end{tikzpicture}
	\caption{This figure shows an example of the first four machines in the order that are used by the \textsc{Delayed-OPT} algorithm. The cost functions belong to the class of 4-step functions. In the figure, $\tilde{c}_1(k)$ is the cost of machine $1$ when the $k^{th}$ step/segment is used but not the $(k+1)^{th}$, and $\cp_4(k)$ is the length of the $k^{th}$ step/segment of machine $4$. The $\phi_j(k)$ values show the order that the \textsc{Delayed-OPT} algorithm fills the segments. In this example, the allocation $A$ of the \textsc{Delayed-OPT} algorithm fully uses the first $6$ segments and also part of the $7^{th}$ segment. The excess of all machines but the third are $0$ and machine $3$ has positive excess since its $2^{nd}$ segment is not fully used.}		
	\label{fig:GendisGuaranteed}
\end{figure}

\begin{theorem} The PoA for the class of 4-step cost functions, assuming two \diss, is constant. 
\label{lem:genDisGuaranteedPoA}
\end{theorem}

 \begin{proof}
 We will show that the total cost of any induced pure Nash equilibrium, assuming {\em two} \diss, is constant away the total cost induced by the \textsc{Delayed-OPT} algorithm which in turns is a constant approximation to the cost of the optimal allocation. In order to show this, we will show that any pure Nash equilibrium $\prof$ has the same allocation with the \textsc{Delayed-OPT} algorithm allocation $A$. This means that for any machine $j$ the load in $\prof$ and $A$ are the same, i.e. $\ell_j(\prof)=\ell_j(A)$.

\begin{claim}
\label{cl:NEequalsA}
For any Nash equilibrium $\prof$, $\ell_j(\prof)=\ell_j(A)$ for all $j$. 
\end{claim}
\begin{proof}
For the sake of contradiction suppose that there exists some Nash equilibrium $\prof$ with different allocation than $A$. 
Then there should be a machine $r$ with $\ell_r(\prof) > \ell_r(A)$. If there are many machines with more load in $\prof$ than in $A$, we choose $r$ to be the one with the maximum $\phi_{r}(\lastSeg_r(\prof))$. 

For any machine $j$, with $\ell_j(\prof) \leq \ell_j(A)$, it holds that the last segment of machine $j$ under $A$ precedes the last segment of machine $r$ under $\prof$ according to segment order $\phi$, i.e. $\phi_{j}(\lastSeg_{j}(A))<\phi_{r}(\lastSeg_r(\prof))$, meaning that overall $\phi_{r}(\lastSeg_r(\prof))$ is the maximum among used segments under $\prof$. The reason is that, if $l$ is the last machine used by the \textsc{Delayed-OPT} algorithm, the excess of all other machines different than $l$ is $0$ under $A$, and therefore if $r\neq l$, $\lastSeg_r(\prof)$ is not used in $A$; by the definition of $\phi$ order, $\phi_{j}(\lastSeg_{j}(A))<\phi_{r}(\lastSeg_r(A)+1) \leq \phi_{r}(\lastSeg_r(\prof))$. If $r=l$, it trivially holds that $\phi_{j}(\lastSeg_{j}(A))<\phi_{r}(\lastSeg_r(A))\leq \phi_{r}(\lastSeg_r(\prof))$.

Next we show that under $\s$ either a regular agent or an \dis\ has an incentive to deviate leading to a contradiction. 

\begin{itemize}
\item If there is at least one \dis\ in machine $r$, or $\lastSeg_{r}(\prof)$ is not full, i.e. $w_r(\prof) \neq 0$, the highest priority regular agent in $r$, $\high_1(S_r(\prof) \cap R)$, is paying $\CumCost_{r}(\lastSeg_r(\prof))$. If this agent deviated to any machine $j$, with $\ell_j(\prof) < \ell_j(A)$ (there exists at least one because $\ell_r(\prof) > \ell_r(A)$), the total load on that machine would be at most $\ell_{j}(A)$ and therefore the agent's payment would be at most  $\CumCost_{j}(\lastSeg_{j}(A))<\CumCost_{r}(\lastSeg_r(\prof))$. 

\item If machine $r$ has only regular agents and $0$ excess, i.e. $w_r(\prof) = 0$, there exists an \dis\ in some machine $j'\neq r$ that is charged with at least $\ve_{j'}(\lastSeg_{j'}(\prof))$. If he deviated to machine $r$, the excess of that machine would become $1$ and he would be the only \dis\ in machine $r$, therefore, he would be charged with $\ve_{r}(\lastSeg_r(\prof)) < \ve_{j'}(\lastSeg_{j'}(\prof))$, where the inequality holds because $\phi_{r}(\lastSeg_r(\prof))$ is the maximum among used segments under $\prof$.
\end{itemize}
\end{proof}
\end{proof}

\begin{theorem} 
The protocol for the class of 4-step cost functions, assuming two \diss, is stable.
\label{lem:genDisGuaranteedStability}
\end{theorem}

Due to space limitations we refer the reader to the appendix for the proof of the theorem.

\section{Resource-Aware Mechanism for Games with Stochastic Arrivals}\label{sec:stoc}

In this section we study the case of stochastic arrivals, where each agent $i$ appears in the system with probability $p_i$ and the mechanism has access to $\mathbf{p}=(p_1,p_2,\dots,p_{|\mathcal{N}|})$. Let $S$ be the random set of the arriving agents and $M$ be a set of machines whose cost 
functions are from the class of 4-step functions. We design a cost sharing scheme with the goal of minimizing the expected price of anarchy defined as follows:

\[\text{ExpectedPoA}(\mathcal G)=\sup_{M, \mathbf{c}\in \mathcal C^{|M|}}\left\{\E_{S\sim P} [\text{PoA}((S,M, \mathbf{c},\Xi))]\,\right\}.\]

Our main theorem (Theorem~\ref{thm:iidpoa}) bounds the expected price of anarchy of our protocol in relation to the expected number of arriving agents
$\tilde{n}=\E_{S\sim\mathbf{p}}[|S|]$. 
For the sake of simplicity in this section we prove the case of identical agents that is $p_i=p$ for all $i$. The proof for the general case (Theorem~\ref{thm:indpoa}) can be found in the appendix.
We show that 
for the case of independently arriving agents there exists a protocol using $3+\lfloor 3\frac{\log({p |\mathcal N|})}{-\log (1-p)}\rfloor$ \diss\
that achieves an expected price of anarchy of at least
$$
\text{ExpectedPoA}(\mathcal G) =O\left(\log \left(\E_{S\sim\mathbf{p}}[|S|] \right)\right) =O(\log \tilde{n}) \, .
$$

\paragraph{Protocol.} 
The protocol is similar to the one we defined with the guaranteed \diss. However rather than using the two guaranteed agents as the \diss\ we choose an appropriate set of \diss\ using the distributional information we have. In order to guarantee stability for any number of \diss\ we adjust the cost sharing protocol by adding two rest points for \diss\ where they pay $0$ share; we further slightly modify the cost shares of \diss\ to include cases where many \diss\ use the same machine. 

 Given the set of arriving agents $S$ and a strategy profile $\prof$ we next define the cost shares of the agents. Let $D\subseteq S$ be the set of \diss\ in $S$ and $R=S\setminus D$ be the set of regular agents in $S$. For simplicity, we drop the dependency on the load and on $\prof$ since there is no ambiguity. The cost share of any \dis\ $i \in \disSet$ using machine $j$ is
\[
	\xi_i (\prof) =
	\begin{cases}
	0 & \text{if } j \in \{1,2\} , w_j \in \{0, \cp_j(\lastSeg_j)-1\} \mbox{ and } \disSet \cap S_j = \{i\}\\
	\ve_{j}(\lastSeg_j) & \text{if } w_j \neq 1, \disSet \cap S_j \supset \{i\}, R \cap S_j \neq \emptyset \mbox{ and } i \in \{\high_1(S_j\cap D), \high_2(S_j\cap D)\}\\
	\ve_{j}(\lastSeg_j-1) & \text{if } w_j = 1, \disSet \cap S_j = \{i\} \mbox{ and } \lastSeg_j >1\\
	\CumCost_{j}(\lastSeg_j) & \text{otherwise.}
	\end{cases}\,
	\]
	
	The cost share of any regular agent $i \in R$ using machine $j$ is (the same as in Section~\ref{sec:genGuaranteedDisruptors})
		
	\[
	\xi_i (\prof) =
	\begin{cases}
	c_j & \text{if } w_j=0, \disSet \cap S_j =\emptyset \mbox{ and } i = \high_1(S_j\cap R) \\
	\CumCost_{j}(\lastSeg_j) & \text{if } w_j=0, \disSet \cap S_j \neq \emptyset \mbox{ and } i = \high_1(S_j\cap R)\\
	\CumCost_{j}(\lastSeg_j) & \text{if } w_j=1 \mbox{ and } i \in \{\high_1(S_j\cap R), \high_2(S_j\cap R)\} \\
	\CumCost_{j}(\lastSeg_j) & \text{if } w_j\notin \{0,1\}  \mbox{ and } i = \high_1(S_j\cap R) \\
	0 & \text{otherwise.}
	\end{cases}\,
	\]

We refer the reader to the appendix for the proof of the following theorem.

\begin{theorem} \label{thm:stablestoc}
The protocol for the class of 4-step cost functions is stable for any number of \diss .
\end{theorem}

Next we continue with upper bounding the expected price of anarchy of our protocol. First we prove two important lemmas on how far the cost of any Nash equilibrium may be from the cost of the allocation $A$ of the \textsc{Delayed-OPT} algorithm, conditioned on the number of \diss\ in the system. We distinguish two cases of having at least three \diss\ or at most two \diss\ in the system. In the first case, the proof is similar to the one of Theorem~\ref{lem:genDisGuaranteedPoA}, but we now need at least three \diss\ because based on the protocol at most two \diss\ may pay $0$ cost shares; those \diss\ have no incentive to deviate to machines that are not used in $A$ and are full with regular agents.

\begin{lemma}\label{lem:highdist}
If $d\geq 3$ \diss\ arrive then the cost of the Nash equilibrium is no more than $d+3$ times the cost of the allocation $A$ of the \textsc{Delayed-OPT} algorithm, by ignoring the arbitrarily small charges of $\ve_{j}(k)$ values. 
\end{lemma}

\begin{proof}
Since we have two rest points for the \diss\ (first case of the cost shares), meaning that at most two \diss\ may pay $0$ in any Nash equilibrium $\prof$, if $d\geq 3$ \diss\ appear in the system, then at least one of them must pay a non-zero share. Following the proof of Theorem~\ref{lem:genDisGuaranteedPoA} we can easily infer that $\prof$ uses the exact same allocation as the \textsc{Delayed-OPT} algorithm (Claim~\ref{cl:NEequalsA}). 
Next we need to bound the overcharging cost. Note that any cost share $\CumCost_{j}(\lastSeg_j(A))$ for some $j$ is no more than the cost of $A$. As a result we simply need to bound the number of agents charged with such a cost share. Let $d_j$ be the number of \diss\ in machine $j$.  

First, we examine the machines other than the last machine used by the \textsc{Delayed-OPT} algorithm. By definition, such machine $j$ will have zero excess, $w_j(A)=0$.   Therefore, if there are only regular agents there will be no overcharging. If there are only \diss\ the overcharging is $d_j  \CumCost_{j}(\lastSeg_j(A))$ which is at most $d_j$ times the total cost of $A$.  If there is at least one \dis\ and at least one regular agent, then we have at most one regular agent paying $\CumCost_{j}(\lastSeg_j(A))$ and either one \dis\ is paying $0$ or two \diss\ are paying the arbitrarily small value $\ve_{j}(\lastSeg_j(A))$. In any case, the overcharging is at most $d_j$ times the total cost of $A$, by ignoring the $\ve_{j}(\lastSeg_j(A))$ values.

Second, let's consider the last machine $r$ used by the \textsc{Delayed-OPT} algorithm. At most two regular agents are charged with $\CumCost_{r}(\lastSeg_r(A))$  (if $w_r(A)=1$). It is also possible that all the \diss\ are charged with $\CumCost_{r}(\lastSeg_r(A))$ and therefore, the total overcharging in machine $r$ is at most $d_r+2$ times the total cost of $A$, by ignoring again the $\ve_{r}(\lastSeg_r(A))$ values. 

Overall, the overcharging is at most  $d+2$ times the total cost of $A$ and as a result, the total cost of $\prof$ is no more than $d+3$ times the total cost of $A$.
\end{proof}

\begin{lemma}\label{lem:lowdist}
If  $r$ regular agents  and $d\leq 2$ \diss\ arrive then the cost of the Nash equilibrium $\prof$ is no more than $r+d$ times the cost of the allocation $A$ of the \textsc{Delayed-OPT} algorithm, by ignoring the arbitrarily small charges of $\ve_{j}(k)$ values. 
\end{lemma}

\begin{proof} 
If the allocation of $\prof$ is not the same as $A$, there must be some machine $j$ such  that $\ell_j(\prof)<\ell_j(A)$. This implies that any agent can deviate to this machine and pay at most $\CumCost_{j}(\lastSeg_j(A))$. As a result the cost of any agent (\dis\ or regular) under $\prof$ is no more than $\CumCost_{j}(\lastSeg_j(A))$ which is upper bounded by the total cost of $A$. Therefore, the total cost of $\prof$ is no more than $r+d$ times the total cost of $A$.
\end{proof}

Next we upper bound the price of anarchy in the special case where the agents arrival probabilities are identical, that is  $p_i=p$ for all $i \in \mathcal N$. Note that in this case the expected number of agents is $\tilde{n}=E_{S\sim\mathbf{p}}[|S|]=p |\mathcal N|$.

\begin{theorem}\label{thm:iidpoa}
For independently arriving agents with identical probabilities $p$ and for the class of 4-step cost functions, there exists a protocol using $|D|= 3+\lfloor 3\frac{\log({p |\mathcal N|})}{-\log (1-p)}\rfloor$ \diss\ that
 achieves an expected price of anarchy of 
$$
\text{ExpectedPoA}(\mathcal G) =O(\log (\E_{S\sim\mathbf{p}}[|S|] )) =O(\log \tilde{n}) \,.
$$
\end{theorem}

Before proceeding with the proof of Theorem~\ref{thm:iidpoa}, we state the following lemma. The proof of the lemma is in the appendix.

\begin{lemma}\label{lem:1}
If we choose a set of \diss\ $D$ such that  $|D|=3+\lfloor 3\frac{\log({p  |\mathcal N|})}{-\log (1-p)} \rfloor$ then 
$$ \Prob[d \leq 2]\E_{S\sim\mathbf{p}}[|S|\mid d\leq 2]\leq 9\, ,$$

where $\Prob$ is the probability symbol and $d=|S\cap D|$ is a random variable depending on $\mathbf{p}.$ 
\end{lemma}

\begin{proof} (Theorem~\ref{thm:iidpoa})
Let $S$ be a random set of arriving agents, $G$ be the corresponding game, $Eq(G)$ be the set of Nash equilibria for $G$ and $A$ be the allocation of the \textsc{Delayed-OPT} algorithm for the set $S$. Moreover, let $d=|S\cap D|$ be a random variable depending on $\mathbf{p}$.
Combining both Lemmas~\ref{lem:highdist} and~\ref{lem:lowdist}  we get that the expected ratio of the cost of the worst case equilibrium to the cost of the allocation of the \textsc{Delayed-OPT} algorithm is 

\begin{equation}
\E_{S\sim\mathbf{p}}\left[ \frac{\max_{\prof\in\textrm{Eq($G$)}}\hat{C}(\prof)}{C(A)}\right] = \Prob[d \leq 2] \E_{S\sim\mathbf{p}}[|S| \mid d\leq 2] +  \Prob[ d\geq 3] \E_{S\sim\mathbf{p}}[d+3\mid d\geq 3] \, .
\label{eq:0}
\end{equation}

We can bound the second summand of Equation~\eqref{eq:0} as

\begin{eqnarray}
\Prob[d\geq 3] \E_{S\sim\mathbf{p}}[d+3\mid d\geq 3]  &\leq& \E_{S\sim\mathbf{p}}[d+3\mid d\geq 3]  \notag\\
&\leq&  \E_{S\sim\mathbf{p}}[d\mid d \geq 0] + 6 =  p |D| +6\, .
\label{eq:1}
\end{eqnarray}

Combining Lemma~\ref{lem:1}, Equation~\eqref{eq:0} and Equation~\eqref{eq:1} we get that

\begin{eqnarray}
\E_{S\sim\mathbf{p}}\left[ \frac{\max_{\prof\in\textrm{E($G$)}}\hat{C}(\prof)}{C(A)}\right] &\leq& p  |D| +15    \leq  \lfloor 3\log (p  |\mathcal N|)  \left(\frac{p}{-\log (1-p)}\right) \rfloor +18 \notag \\
&\leq& 3\log (p  |\mathcal N|)  \left(\frac{p}{-\log (1-p)}\right) +18
\notag \\
&\leq& 3\log (p  |\mathcal N|) + 18 = 3\log \tilde{n} + 18\, , 
\end{eqnarray}
where the last inequality is due to $(\frac{p}{-\log (1-p)})\leq 1$ for $p\geq 0$. The fact that the cost of the \textsc{Delayed-OPT} algorithm outcome is a constant approximation to the optimum cost completes the proof.
\end{proof}

Next we upper bound the price of anarchy when the agents arrival probabilities are not necessarily identical. For $p_i$ being the probability that agent $i$ arrives, $\tilde{n}=E_{S\sim\mathbf{p}}[|S|]=\sum_{i=1}^{|\mathcal{N}|} p_i$ is the expected number of agents. The  proof of the theorem is in the appendix.

\begin{theorem}\label{thm:indpoa}
For independently arriving agents with not necessarily identical probabilities and for the class of 4-step cost functions, there exists a protocol that
 achieves an expected price of anarchy of 
$$
\text{ExpectedPoA}(\mathcal G) =O(\log (\E_{S\sim\mathbf{p}}[|S|] )) =O(\log \tilde{n}) \,.
$$
\end{theorem}

\subsection*{Acknowledgments}
The authors would like to thank Giorgos Christodoulou for helpful discussions during the initial stages of this project. The work of the first author was partially supported by NSF CAREER award CCF-2047907.	

\bibliographystyle{abbrvnat}
\bibliography{cost-sharing}

\appendix

\section{Additional Discussion Regarding our Results}
\label{sec:appendixdisc}

In this section we briefly address some more technical aspects regarding our results and its comparison to prior work. We first discuss the type of overcharging that we use in our mechanisms, and compare it to the types of overcharging that has been used in prior work. We then provide some (partial) justification and intuition behind some limits that we require on the number of \diss, and the capacities of the capacitated constant cost functions.

\paragraph{Power of overcharging.}
At the core of our cost-sharing mechanism lies the ability of \diss\ to penalize other agents through overcharging. In fact, the total amount of overcharging that our mechanism enforces can depend not only on the number of agents using a machine, but also on the actual set of agents (e.g., on whether one of them is an \dis\ or not). This is in contrast to the way that some prior cost-sharing mechanisms have used overcharging, e.g., in \cite{CGS17,CGLS20}. In fact, the result of \citet{CGS17}, showing that no resource-aware mechanism can achieve a PoA better than $O(\sqrt{n})$, even if it uses overcharging, assumes that the amount of overcharging would depend only on the number of users, not the set of users of the corresponding machine. In light of this observation, one could argue that our mechanisms in this paper do not only leverage the additional information that they have, relative to prior-free resource-aware mechanisms; they actually also leverage the ability to introduce overcharging in a more flexible way.

\paragraph{Restriction on the number of \diss.}
Assume that there was only one \dis\ appearing in the system. Consider a number of agents such that in the allocation of the \textsc{Delayed-OPT} algorithm the last machine $r$ is fully used, i.e., $n_r=\cp_r$ agents are allocated in machine $r$. The \dis\ should not have an incentive to deviate to machines full of regular agents and therefore, he should pay at most $\ve_{r-1}$. Moreover, he should pay more than $\ve_{r+1}$ so he deviates to machine $r+1$ if it is full of regular agents. This is the reason we assume the monotonicity on the $\ve_j$ values and we allocate the \dis\ in machine $r$. 

Consider now a number of agents such that in the allocation of the \textsc{Delayed-OPT} algorithm machine $r$ is 
used by $n_r=\cp_r-1$ agents. In the stable outcome, if we allocate the \dis\ to any full machine $j$, he has an incentive to deviate to the last machine, make it full and pay $\ve_r < \ve_{j}$. Therefore, the only option is to allocate the \dis\ to the last machine $r$. In order for such allocation to be an equilibrium, he should be charged with some $\ve'_r <\ve_{r-1}$ so he has no incentive to deviate to prior machines full of regular agents. By considering now $n_r=\cp_r-2$ and using the same arguments we can show that if the \dis\ uses machine $r$, he should be charged with something strictly less than $\ve_{r-1}$. Similarly, for any $n_r$ we can show that the \dis\ should be charged with an amount strictly less than $\ve_{r-1}$, and therefore the same should hold for $n_r=1$. But for $n_r=1$ the \dis\ would be alone in $r$ and hence, he should pay at least $c_r$  which leads to a contradiction. Considering \diss\ in pairs resolves this issue and results in the existence of stable outcomes.

\paragraph{Restriction on the capacities.}
This was derived from the need to distinguish between the two cases in the proof of Lemma~\ref{lem:disGuaranteedStability}. First notice that the two \diss\ should be together, otherwise they have an incentive to deviate to a machine full of regular agents. Suppose now that $n_r=2$ agents use the last machine $r$ based on the allocation of the \textsc{Delayed-OPT} algorithm. If both \diss\ were placed in machine $r$, one of them should be charged with at least $c_r/2$ and still have an incentive to deviate to a full machine. Therefore, for $n_r\leq 2$, in any equilibrium the \diss\ do not use $r$.

One the other hand, in the case that $n_r=\cp_r-1$, if no \dis\ was placed in machine $r$, the lowest priority agent $i$ in machine $j$ that the \diss\ use, may have an incentive to deviate to machine $r$. This could be the case because agent $i$ is currently charged with $\CumCost_j$, which could possibly be greater than $c_r$, where $c_r$ would be the maximum charge of agent $i$ if he deviated to machine $r$; the reason is that $r$ would become full. This means that for $n_r=\cp_r-1$, in any equilibrium at least one \dis\ should use machine $r$ (this automatically means that $\cp_r$ should be different than $1$). 

Overall, the above two points indicate the restriction of $\cp_r-1>2$, which results in the requirement of $\cp_j\geq 4$ for all machines $j$.

\section{Missing Proofs from Section~\ref{sec:guar}}
\subsection{Proof of Theorem~\ref{lem:genDisGuaranteedStability}}

We next show how we construct a strategy profile that is a Nash equilibrium.
Let $A$ be the allocation 
 according to the \textsc{Delayed-OPT} algorithm and let $r$ be the last machine used by the \textsc{Delayed-OPT} algorithm. We will consider cases based on the excess of that machine, $w_{r}(A)$, and its load, 
$\ell_{r}(A)$.  For any other machine $j\neq r$ note that $j$ has $0$ excess, i.e. $w_{j}(A)=0$. 

 For the construction of the Nash equilibrium we allocate the agents the same way as the \textsc{Delayed-OPT} algorithm does, but we need to carefully assign the \diss\ and regular agents to the appropriate machines in order to ensure stability. The following assignment guarantees stability.

\begin{itemize}
\item If only one machine is used by the \textsc{Delayed-OPT} algorithm, then simply allocating any agent to this machine would be Nash equilibrium. 

\item If more than one machine is used by the \textsc{Delayed-OPT} algorithm, we distinguish between two cases based on the values of $w_{r}(A)$ and $\ell_{r}(A)$.

\begin{itemize}
\item $w_{r}(A)\neq 1$ and $\ell_{r}(A)>2$: We allocate the two \diss\ to machine $r$ and the regular agents in a way such that the allocation coincides with $A$ by ensuring that agents paying non-zero cost shares are the highest priority agents, with the lowest among them to be allocated to machine $r$.

  Obviously, regular agents paying $0$ have no incentive to deviate. If an \dis\ from machine $r$ deviated to another machine $j$ then the excess would increase to $1$ and the \dis\ would pay $\ve_{j}(\lastSeg_j(A))$, but ,by the definition of the $\ve_j(k)$ values, this is higher than his current payment of $\ve_{r}(\lastSeg_r(A))$. Finally, regarding regular agents with non-zero cost shares, if they deviated to machine $r$ they would have the highest priority and pay at least $\CumCost_{r}(\lastSeg_r(A))$ which is more than their current charge. If they deviated to any machine $j\neq r$, they would increase the excess to $1$ and pay $\CumCost_{j}((\lastSeg_j(A)+1))> \CumCost_{r}(\lastSeg_r(A))$. 
  
\item $w_{r}(A) = 1$ or $\ell_{r}(A)=2$: We allocate the two \diss\ to machine $r'\neq r$ that is the last machine used by the \textsc{Delayed-OPT} algorithm before $r$.\footnote{We do not assign the \diss\ to machine $r$, because in the cases that the excess is $1$ and there are two \diss\ or there is no regular agents, the \diss\ pay a high cost.} 
We allocate the regular agents in a way such that the allocation coincides with $A$ by ensuring that agents paying non-zero cost shares are the highest priority agents with the two lowest among them (or the one lowest, if $\ell_{r}(A) = 1$) to be allocated to machine $r$.

  Similarly, regular agents paying $0$ have no incentive to deviate. If an \dis\ from machine $r'$ deviated to another machine $j\neq r$ the same argument as before holds. If he deviated to machine $r$, he would be the only \dis\ and the excess of $r$ would become either $2$ or $3$, meaning that his charge would be $\CumCost_{r}(\lastSeg_r(A))$. Finally, regarding regular agents with non-zero cost shares, the same arguments as before hold. 
\end{itemize}

\end{itemize}

\section{Missing Proofs from Section~\ref{sec:stoc}}
\subsection{Proof of Theorem~\ref{thm:stablestoc}} 
We  show how we construct a strategy profile that is a Nash equilibrium similarly to Theorem~\ref{lem:genDisGuaranteedStability}. Let $A$ be the allocation  
 according to the \textsc{Delayed-OPT} algorithm and let $r$ be the last machine used by the \textsc{Delayed-OPT} algorithm. We will consider cases based on the excess of that machine, $w_{r}(A)$ and its load, 
$\ell_{r}(A)$.  

If only one machine is used in the \textsc{Delayed-OPT} algorithm then simply allocating any agent to this machine would be Nash equilibrium. Next we consider separately the cases of no \dis , one \dis\ and at least two \diss .

\begin{itemize}
\item {\bf Only regular agents}
\begin{itemize}
\item If  $w_r(A)\neq \cp_r(\lastSeg_r(A))-1$, we assign agents according to $A$ making sure that the highest priority agents are responsible for non-zero costs, i.e. if an agent has a cost share of $0$ then all lower priority agents have a cost share of $0$. Additionally we ensure that among the agents paying non-zero cost, the lowest priority agents are allocated to machine $r$; in other words, if $w_r(A)=1$ we ensure that the two highest priority agents in machine $r$ have the lowest priority among agents with non-zero cost shares and if $w_r(A)\neq 1$ we ensure that the one highest priority agent in machine $r$ has the lowest priority among agents with non-zero cost shares.
 
Every  agent  with $0$ cost share has no incentive to deviate. Consider some agent with non-zero cost share. Deviating to some machine $j\neq r$ would increase the excess to $1$ causing the two highest priority agents to pay a non-zero cost share equal to $\CumCost_{r}(\lastSeg_r(A)+1)>\CumCost_{r}(\lastSeg_r(A))$. Currently there is only one agent charged with non-zero cost share in machine $j$, and therefore the deviating agent would have one of the two highest priorities which means that the deviation is not profitable. If the agent deviated to $r$, since $w_r(A)\neq \cp_r(\lastSeg_r(A))-1$  and the agent has higher priority than any agent in machine $r$, he would pay either $\CumCost_{r}((\lastSeg_r(A)+1))$ if $w_r(A)= 0$ or $\CumCost_{r}(\lastSeg_r(A))$  otherwise. In either case this is not a desirable deviation.

\item If $w_r= \cp_r(\lastSeg_r(A))-1$ then we may need to alter the allocation $A$  as follows. Let $j$ be the machine with the highest total cost $c_j(A)$ according to $A$. We move one agent from machine $j$ to machine $r$ (or do nothing if $j=r$). Similarly as above, we make sure that the highest priority agents are responsible for non-zero costs and the lowest priority of those are placed in machine $j$.  

We claim that this assignment is a Nash equilibrium. Naturally agents that have $0$ cost share have no incentive to deviate. If an agent with non-zero cost deviates to some machine $j'\neq j$  will result in excess of $1$ and will pay 
$\CumCost_{j'}(\lastSeg_{j'}(A)+1)> \CumCost_{r}(\lastSeg_r(A))$. Deviating to machine $j$ will result in $0$ excess and since the agent has the highest priority will pay the total cost of the machine. But since this agent was already paying the total cost of another machine and the cost of machine $j$ is the highest the deviation is not profitable.
\end{itemize}

\item {\bf Exactly one \dis }
\begin{itemize}
\item  If  $w_r(A)\neq \cp_r(\lastSeg_r(A))-1$, we put the \dis\ to either machine $1$ or $2$ depending on which of them has excess of $0$ (always one of them has) so that the \dis\ pays $0$. We allocate arbitrarily the rest of the agents according to $A$ making sure that the highest priority agents are responsible for non-zero costs, ensuring that the lowest priority agents among them are allocated to machine $r$.

 Since the machine that the \dis\ occupies (machine $1$ or $2$) has excess $0$ the \dis\ pays $0$ and has no incentive to deviate. The same arguments as in the case of only regular agents can be used here in order to show that no regular agent has an incentive to deviate.

\item If $w_r(A)= \cp_r(\lastSeg_r(A))-1$ then we allocate the \dis\ to either machine $1$ or $2$ depending on which of them has excess of $0$; let $f$ be that machine. We first assign the regular agents according to $A$ and then adjust the assignment as follows. Let $j$ be the machine with the maximum total cost $c_j(A)$ according to $A$. If $\CumCost_{f}(\lastSeg_f(A))>c_j(A)$ we change $j$ to be $f$. Then we move one agent from machine $j$ to machine $r$, unless $j=r$, and similarly as above, we make sure that the highest priority agents are responsible for non-zero costs and the lowest priority among them are placed in machine $j$.  

The \dis\ pays $0$ since machine $f$ has either excess $0$ or excess $\cp_{f}(\lastSeg_f(A)-1)$  and therefore the \dis\ has no incentive to deviate. If any agent deviated to some machine $j'\neq j$ would result in a cost $\CumCost_{j'}(\lastSeg_{j'}(A)+1)> \CumCost_{r}(\lastSeg_r(A))$. Deviating to machine $j$ would result in a cost share equal to either $\CumCost_{f}(\lastSeg_f(A))$ if $j=f$ or $c_j(A)$ otherwise. Since we picked $j$ such that the cost share of deviating to be the maximum, the deviation is not profitable.
\end{itemize}

\item {\bf At least two \diss}\\
Let $d$ be the number of \diss\ that arrive and $f$ be either machine $1$ or $2$ as long as it is different from $r$; $f$ has excess of $0$. We allocate the agents according to $A$ and we assign the spots to \diss\ and regular agents as follows:

\begin{itemize}
\item $w_r(A)\neq 1$ and $\ell_r(A)>2$\\
We allocate the \diss\ in pairs, following their priority order from high to low, to machines with increasing $\ve_{j}(\lastSeg_j(A))$, starting from machine $r$; in the case that $d$ is odd, we assign only one \dis\ to machine $f$. If there are more \diss , we allocate one more \dis\ based on the priority order to machine $f$ in the case that $d$ is odd and the rest of the \diss\ are allocated arbitrarily.
\item $\ell_r(A)\leq 2$\\
Similarly as above, we allocate the \diss\ in pairs as above but by ignoring machine $r$. If there are more \diss\ we allocate the next pair\footnote{Note that there would be a pair of \diss\ because if $d$ is odd, we would have allocated odd number of \diss\ so far and if $d$ is even we would have allocated even number of \diss\ so far.} to machine $r$ if $\ell_r(A) = 2$ or the next \dis\ to machine $r$ if $\ell_r(A) =1$. If there are more \diss\ we allocate the next one to machine $f$ if it has only one \dis\ and the rest arbitrarily. 
\item $w_r(A) = 1$ and $\ell_r(A)>2$\\
We start allocating the \diss\ as in the second case with the only difference that we if we allocate a pair of \diss\ in machine $r$, we make sure that the {\em highest} priority pair of \diss\ is allocated in $r$. 
\end{itemize}

We allocate the regular agents arbitrarily making sure that the highest priority regular agents are charged with non-zero cost shares, ensuring that machine $r$ has the lowest priority among them. 

Obviously, agents paying $0$ have no incentive to deviate. 

If an \dis\ deviated to machine $r$ and the first case applies where $w_r(A)\neq 1$ and $\ell_r(A)>2$, then machine $r$ has already the two highest priority \diss\ and therefore the deviating \dis\ would pay either  $\CumCost_{r}(\lastSeg_r(A)+1)$  or $\CumCost_{r}(\lastSeg_r(A))$  both of which are not profitable. If an \dis\ deviated to $r$ and $\ell_r(A)\leq 2$ then either he would be the only \dis\ or $r$ had no regular agent; in both cases the deviating \dis\ would pay $\CumCost_{r}(\lastSeg_r(A))$. If an \dis\ deviated to $r$ and $\ell_r(A)> 2$ but $w_r(A) = 1$, the excess of $r$ would become $2$ and machine either $r$ has already the two highest priority \diss\ or it has no \dis , meaning in both cases that the deviating \dis\ would pay $\CumCost_{r}(\lastSeg_r(A))$.

If an \dis\ deviated to another machine $j$ then the excess increases to $1$. If there was already at least one \dis\ then the deviating \dis\ would pay $\CumCost_{j}(\lastSeg_j(A)+1)>\CumCost_{r}(\lastSeg_r(A))$. If there was no other \dis\ then the deviating \dis\ would pay $\ve_{j}(\lastSeg_j(A))$ but by our assignment this must be higher than the current payment of the \dis. 

Finally if any regular agent with non-zero cost share deviated to $r$, he would be the highest priority agent and pay $\CumCost_{r}(\lastSeg_r(A)+1)$ or $\CumCost_{r}(\lastSeg_r(A))$ and if he deviated to $j\neq r$  he would increase the excess to $1$ and pay $\CumCost_{j}(\lastSeg_j(A)+1)> \CumCost_{r}(\lastSeg_r(A))$. 
 \end{itemize}

\subsection{Proof of Lemma~\ref{lem:1}}
We consider two cases with respect to the expected number of agents.

If $p|\mathcal{N}|\leq 1$, $$ \Prob[d \leq 2] \E_{S\sim\mathbf{p}}[|S| \mid d\leq 2]\leq \E_{S\sim\mathbf{p}}[|S| \mid d\leq 2] \leq p |\mathcal N\setminus D|+2 \leq p  |\mathcal N|+2\leq 3\, .$$

If $p|\mathcal{N}|> 1$,
\begin{eqnarray*}
 \Prob[d \leq 2] \E_{S\sim\mathbf{p}}[|S| \mid d\leq 2]]  &\leq&
 ( \Prob[d = 0]+\Prob[d = 1]+\Prob[d = 2])(p |\mathcal N\setminus D|+2)\\
  &\leq& ((1-p)^{|D|}+(1-p)^{|D|-1}p|D|+(1-p)^{|D|-2}p^2|D|^2)(p|\mathcal N|+2)\\
  &\leq&(1+p|D|+p^2|D|^2)(1-p)^{|D|-2}(p|\mathcal N|+2)\\
  &\leq& 3(p|\mathcal N|)^{2}(1-p)^{|D|-2}(p|\mathcal N|+2)\\
  &\leq& 9(p|\mathcal N|)^{3}(1-p)^{|D|-2}\\
  &\leq& 9(p|\mathcal N|)^{3}(1-p)^{\frac{\log\left(\frac{1}{(|\mathcal N| p)^3}\right)}{\log (1-p)}}\\
  &=& 9(p|\mathcal N|)^{3}\frac{1}{(p |\mathcal N|)^3} = 9.
\end{eqnarray*}

where the last inequality comes from the fact that  $|D|-2=1+\lfloor \frac{\log\left(\frac{1}{(|\mathcal N| p)^3}\right)}{\log (1-p)}\rfloor \geq \frac{\log\left(\frac{1}{(|\mathcal N| p)^3}\right)}{\log (1-p)}$

\subsection{Proof of Theorem~\ref{thm:indpoa}}
Let $S$ be a random set of arriving agents, $G$ be the corresponding game, $Eq(G)$ be the set of Nash equilibria for $G$ and $A$ be the allocation of the \textsc{Delayed-OPT} algorithm for the set $S$. Moreover, let $d=|S\cap D|$ be a random variable depending on $\mathbf{p}=(p_1,p_2,\ldots ,p_{|\mathcal{N}|})$.
Similarly to the proof of Theorem~\ref{thm:iidpoa}, combining both Lemmas~\ref{lem:highdist} and~\ref{lem:lowdist}  we get that the expected ratio of the cost of the worst case equilibrium to the cost of the allocation of the \textsc{Delayed-OPT} algorithm is

\begin{equation}
\E_{S\sim\mathbf{p}}\left[ \frac{\max_{\prof\in\textrm{Eq($G$)}}\hat{C}(\prof)}{C(A)}\right] = \Prob[d \leq 2] \E_{S\sim\mathbf{p}}[|S| \mid d\leq 2] +  \Prob[ d\geq 3] \E_{S\sim\mathbf{p}}[d+3\mid d\geq 3] \, .
\label{eq:ind0}
\end{equation}
 
W.l.o.g. assume $p_1\geq p_2\geq \dots \geq p_{|\mathcal{N}|}$. Then we designate the set of \diss\ $\disSet$ to be the agents associated with the highest probabilities such that $\sum_{i=3}^{|D|} p_i =3\log \tilde{n}$ (we always designate as \diss\ the two agents with the highest probability). We first bound the probability $\Prob[d \leq 2]$ for the case that $\tilde{n}> 1$ as follows.

\begin{eqnarray}
\Prob[d \leq 2] &=& \Prob[d = 0]+\Prob[d = 1]+\Prob[d = 2] \notag\\
&=&  \prod_{i=1}^{|D|} (1- p_i) + \sum_{j=1}^{|D|} p_j \cdot \prod_{i=1, i\neq j}^{|D|} (1- p_i) + \sum_{k=1}^{|D|} p_k \cdot \sum_{j=1,j\neq k}^{|D|} p_j \cdot \prod_{i=1, i\notin \{j,k\}}^{|D|} (1- p_i)\notag\\
&\leq& \prod_{i=3}^{|D|} (1- p_i) \left(1+\sum_{j=1}^{|D|}p_j + \sum_{k=1}^{|D|} p_k \cdot \sum_{j=1,j\neq k}^{|D|} p_j\right)\notag\\
&\leq& \left( \frac{\sum_{i=3}^{|D|}(1-p_i) }{|D|-2}   \right)^{|D|-2} \left(1+\sum_{j=1}^{|\mathcal N|}p_j + \sum_{k=1}^{|\mathcal N|} p_k \cdot \sum_{j=1}^{|\mathcal N|} p_j\right)\notag\\
&\leq& \left( 1-\frac{\sum_{i=3}^{|D|}p_i }{|D|-2}   \right)^{|D|-2} \cdot 3 \tilde{n}^2 = 3\tilde{n}^2\left( 1-\frac{3\log \tilde{n}}{|D|-2}   \right)^{|D|-2}   \notag\\
&\leq& 3\tilde{n}^2e^{-3\log \tilde{n}} = 3\tilde{n}^2\frac 1{\tilde{n}^3} = \frac 3{\tilde{n}} \notag\, ,
\end{eqnarray}
where the second inequality comes from the AM/GM inequality: $\left(\prod_{i=1}^k x_i\right)^{1/k}\leq  1/k\sum_{i=1}^k x_i $, and the final inequality follows from the fact that $(1-x/k)^k\leq e^{-x}$.   As a result, we can bound the first summand of~\eqref{eq:ind0} by $9$ as follows. 

If $\tilde{n}\leq 1$,  then 

\begin{equation}
\Prob[d \leq 2] \E[|S| \mid d \leq 2] \leq \E[|S| \mid d \leq 2] \leq \tilde{n}+2 \leq 3 \, ,
\label{eq:ind1}
\end{equation}

otherwise,

\begin{equation}
 \Prob[d \leq 2] \E[|S| \mid d \leq 2] \leq \frac{3}{\tilde{n}}  (\tilde{n}+2)\leq 9\, .
\label{eq:ind2}
\end{equation}

Similarly to Equation~\eqref{eq:1} we can bound the second summand of ~\eqref{eq:ind0}  as
\begin{equation}
\label{eq:ind3}
\Prob[ d\geq 3]  \E[d+3\mid d\geq 3] \leq \E[d\mid d\geq 0 ] +6 \leq \sum_{i=1}^{|D|} p_i + 6\leq \sum_{i=3}^{|D|} p_i  + 8= 3\log \tilde{n} +8
\end{equation}

Combining Equations~\eqref{eq:ind0},~\eqref{eq:ind1},~\eqref{eq:ind2} and~\eqref{eq:ind3} completes the proof of the theorem.

\end{document}